%% file: IPDPS_main.tex
\documentclass[conference]{IEEEtran}
\usepackage{amsmath,amssymb,amsfonts,amsthm}
\usepackage{color}
\usepackage{graphicx}
\usepackage{epstopdf}
\usepackage{bbm}
\usepackage{algpseudocode} %In fact part of the "algorithmicx" package.
\usepackage{algorithm}
\usepackage{subcaption} %To put subfigures in a figure environement

\newtheorem{lemma}{\textbf{Lemma}}

\newtheorem{definition}{\textbf{Definition}}
\newtheorem{thm}{\textbf{Theorem}}
\newtheorem{clm}{\textbf{Claim}}
\newtheorem{fact}{\textbf{Fact}}

\newtheorem{example}{\textbf{Example}}

\newtheorem{remark}{\textbf{Remark}}

\newtheorem{question}{\textbf{Question}}
\newcommand{\alic}[1]{{\color{red}[\textbf{Ali: }\emph{#1}]}}

\newcommand{\ie}{{i.e.}}
\newcommand{\eg}{{e.g.}}

\renewcommand{\S}{Section~}

\newcommand{\mc}[1]{\mathcal{#1}}

\renewcommand{\Pr}[1]{\ensuremath{\operatorname{\mathbf{Pr}}\left[#1\right]}}
\newcommand{\Ex}[1]{\ensuremath{\operatorname{\mathbf{E}}\left[#1\right]}}
\newcommand{\Var}[1]{\ensuremath{\operatorname{\mathbf{Var}}\left[#1\right]}}

\title{Proximity-Aware Balanced Allocations \\ in Cache Networks}

\author{
    \IEEEauthorblockN{Ali Pourmiri\IEEEauthorrefmark{1}, Mahdi Jafari Siavoshani\IEEEauthorrefmark{2}, Seyed Pooya Shariatpanahi\IEEEauthorrefmark{1} }
    \IEEEauthorblockA{\IEEEauthorrefmark{1}School of Computer Science\\
    Institute for Research in Fundamental Sciences (IPM),
    Tehran, Iran
    \\\{pourmiri, pooya\}@ipm.ir}
    \IEEEauthorblockA{\IEEEauthorrefmark{2} Department of Computer Engineering\\
    Sharif University of Technology, Tehran, Iran\\
    {mjafari}@sharif.edu}
} 
\thanks{This work was funded in part by ...
    }

\begin{document}
\maketitle

%--------------------------------------------------------------------
% Abstract ----------------------------------------------------------
\input{1-Abstract}

\centerline{\textbf{Keywords}}
Randomized Algorithms, Distributed Caching Servers, Request Routing, Load Balancing,  Communication Cost, Balls-into-Bins, Content Delivery Networks.

%--------------------------------------------------------------------
%Section-------------------------------------------------------------
\input{2-Introduction}

%--------------------------------------------------------------------
%Section-------------------------------------------------------------
\input{3-NotationProbSetting}

%--------------------------------------------------------------------
%Section-------------------------------------------------------------
\input{4-MinCommCostStrategy}

%--------------------------------------------------------------------
%Section-------------------------------------------------------------
\input{5-PowTwoChoiceStrategy}

%--------------------------------------------------------------------

%Section-------------------------------------------------------------
\input{7-Simulations}

%Section-------------------------------------------------------------
\input{9-Discussions}

%Section-------------------------------------------------------------
\section*{Acknowledgment}
The authors would like to thank Seyed~Abolfazl~Motahari, Omid Etesami, Thomas~Sauerwald and Farzad~Parvaresh for helpful discussions and feedback.

% Bibliography-------------------------------------------------------
\bibliographystyle{IEEEtran}
\bibliography{../diss1}
%\input{bib.tex}
%\bibliographystyle{IEEEtranS}
%\bibliography{diss1}

\appendices
\clearpage
\input{8-Appendix}

 \end{document}

%% file: 1-Abstract.tex
\begin{abstract}
We consider load balancing in a network of caching servers delivering contents to end users.
 Randomized load balancing via the so-called \emph{power of two choices} is a well-known approach in parallel and distributed systems that reduces network imbalance. In this paper, we propose a randomized load balancing scheme which simultaneously considers cache size limitation and proximity in the server redirection process.

Since the memory limitation and the proximity constraint cause correlation in the server selection process, we may not benefit from the power of two choices in general.  
  However, we prove that in certain regimes, in terms of memory limitation and proximity constraint, our scheme results in the maximum load of order $\Theta(\log\log n)$ (here $n$ is the number of servers and requests), and at the same time, leads to a low communication cost. This is an exponential improvement in the maximum load compared to the scheme which assigns each request to the nearest available replica. 
Finally, we investigate our scheme performance by extensive simulations.  
\end{abstract}

%% file: 2-Introduction.tex
\section{Introduction}\label{sec:Introduction}
\subsection{Problem Motivation}
Advancement of technology leads to the spread of smart multimedia-friendly communication devices to the masses which causes a rapid growth of demands for data communication \cite{Cisco_Report}. Although
Telcos have been spending hugely on telecommunication infrastructures,  they cannot keep up with this data demand explosion.
Caching predictable data in network off-peak hours, near end users, has been proposed as a promising solution to this challenge. This approach has been used extensively in content delivery networks (CDNs) such as Akamai, Azure, Amazon CloudFront, etc. \cite{Zhang2013}, \cite{NygrenSS10}, and mobile video delivery  \cite{GolrezaeiHelper}. In this approach, a \emph{cache network} is usually referred to as a set of caching servers that are connected over a network, giving content delivery service to end users.

In cache networks, load balancing is one of the most important challenges when assigning requests to servers. This assignment strategy is implemented either at network-side or client-side. In the first approach there is a centralized authority which maps requests to servers, while balancing out the load. This authority employs network status information to optimally allocate requests to servers, which often involves complex algorithms. However, in the latter, the clients choose their favorite servers in a distributed fashion. In this paper we focus on the distributed server selection approach.

Randomized load balancing via the so-called ``power of two choices'' is a well-investigated paradigm in parallel and distributed settings \cite{ABKU99,Mitzenmacher01,AdlerCMR98,LenzenW16}. In this approach, upon arrival of a request, the corresponding user will query about current load of two independently at random chosen  servers, and then allocates the request to the least loaded server.
Berenbrink et al. \cite{BCSV06} showed that in this scheme after allocating  $m$ balls (requests, tasks, ...) to $n$ bins (servers, machines, ...) the   maximum number of balls assigned to any bin, called {\it maximum load}, is at most $m/n+O(\log\log n)$ with high probability. This  only deviates $O(\log\log n)$ from the average load and the deviation  depends on the number of servers. However, in many settings, selecting any two random servers  might be infeasible or costly. For example \emph{proximity principle} in CDNs for server selection is essential to reduce communication cost; \ie, each request should be redirected to a nearby server.	
	
Considering this constraint, Kenthapadi and Panigrahi \cite{KP06} proposed a model where $n$ bins are organized as a $d$-regular graph. Corresponding to each ball, a node is chosen uniformly at random as the first candidate. Then, one of its neighbours is chosen uniformly at random as the second candidate and the ball is allocated to the one with the minimum load. 
%Corresponding to each ball, one edge is chosen uniformly at random  and  the ball is then allocated to one of the edge endpoints with minimum load. 
Under this assumption, they proved that if the graph is sufficiently dense (\ie, the average degree is $n^{\Omega(\log\log n/\log n)}$), then after allocating  $n$ balls the maximum load is $\Theta(\log\log n)$ with high probability.
Although the model fairly considers the proximity principle, due to the cache limitation it cannot be directly applied  in cache networks.

In summary, the proximity principle can be in tension with load balancing in many situations, as nearby users may be congested. This leads to a fundamental trade-off between the maximum load and the communication cost.
 Hence, designing a distributed assignment strategy to handle this trade-off optimally is a central and challenging goal in cache networks.

\iffalse
Due to cache limitation, cache networks are 

 Although the model  fairly exhibits the effect of power of two correlated choices, due to cache limitation, any request cannot be handled by  nearby caching servers. Hence,
	proximity principle in cache networks can be in tension with load balancing, and. }
\fi

\subsection{Problem Setting and Our Contributions}
While many authors have used the idea of power of two choices in server-selection algorithms, theoretical foundations of this phenomena in the context of cache networks with communication cost, has not yet been investigated. In this paper, we consider a general cache network model that entails basic characteristics of many practical scenarios. We consider a grid network of $n$ servers, each equipped with a cache of size $M$. Also there are $n$ sequential file requests, from a library of size $K$, distributed among servers uniformly at random.
  %Thus, each server is responsible for approximately $\mathrm{Po}(1)$ file requests\footnote{$\mathrm{Po}(\lambda)$ here stands for a random variable generated according to Poisson distribution with parameter $\lambda$.} from a library of size $K$. 
Let us assume a popularity distribution $\mathcal{P}=\{p_1,\dots,p_K\}$ for the library. We assume cache placement at each server is proportional to this popularity distribution. Every server either serves its requests or redirects them (via an assignment scheme) to other nodes which have cached the files. We define the maximum load of an assignment scheme as the maximum number of allocations to any single server after assigning all requests. The communication cost is the average number of hops required to deliver requested file to its request origin. 

In the simplest assignment scheme, each request arrived at every server should be dispatched to the nearest file replica. This scheme results in the minimum communication cost, while ignoring maximum load of servers. We show that, for every  constant $0<\alpha<1/2$, if  $K=n$, $M=n^{\alpha}$, and $\mathcal{P}$ is a uniform distribution, this scheme will result in the maximum load in the interval $[\Omega(\log n/\log\log n), O(\log n)]$ with high probability\footnote{With high probability refers to an event that happens with probability $1-1/n^c$, for some constant $c>0$.} (w.h.p.).  Moreover, for
every constant $0<\epsilon<1$, if 
 $K=n^{1-\epsilon}$ and $M=\Theta(1)$, then  the maximum load is $\Theta(\log n)$ w.h.p.
We also investigate the communication cost occurred in this scheme for Uniform and Zipf popularity distributions. In particular, we derive the communication cost of $\Theta(\sqrt{K/M})$ for the Uniform distribution.

%In the second scheme, for every request we randomly pick two servers which have cached the requested file, and  the request is then assigned to the one with lesser load. We investigate the role of memory limitation in load balancing performance in different network regimes, in terms of $K$, $n$, and $M$. Basically, we show that due to memory limitation, the two chosen servers will become correlated, and this might diminish the power of two choices. However, if the file replication $nM/K$ is of order $n^{\Omega(1)}$, we prove that the maximum load is of order $\Theta(\log \log n)$.

In contrast, we propose a new scheme which considers both maximum load and communication cost objectives simultaneously. For each request, this scheme chooses two random candidate servers that have cached the request while putting a constraint on their distance $r$ to the requesting node (\ie, the proximity constraint). Due to cache size limitation and the proximity constraint, current results in the balanced allocation literature cannot be carried over to our setting. Basically, we show that here the two chosen servers will become correlated and this might diminish the power of two choices. Since this correlation arises from both memory limitation and proximity principle, the main challenge we address in this paper is characterizing the regimes where we can benefit from the power of two choices and at the same time have a low communication cost.

In particular, 
	suppose $0<\alpha, \beta<1/2$ be two constants and
	let  $K=n$, $M=n^\alpha$,  $r=n^{\beta}$, and $\mathcal{P}$ be a Uniform distribution. Then, provided $\alpha+2\beta\geq 1+2(\log\log n/\log n)$, the maximum load is $\Theta(\log\log n)$ w.h.p., and the communication cost is $\Theta(r)$.
Therefore, we deduce that if we set  $M=n^\alpha$, for some constant $0<\alpha<1/2$, then it is sufficient to have $\beta={\frac{1-\alpha}{2}}+\log\log n/\log n$ and hence $r=n^{\frac{1-\alpha}{2}} \log n$. This means that the communication cost is only $\log n$ factor above the communication cost achieved by the nearest replica strategy, which is $\Theta(\sqrt{K/M})=\Theta({n^{\frac{1-\alpha}{2}}})$.

\iffalse	
	\alic{In this paper, we consider two main memory regimes, $M=\Theta(1)$ and $M=n^{\epsilon}$, where $0<\epsilon<1/2$ is a constant, which we call them {\emph low} and {\emph high} memory regimes, receptively. Notice that, to have a sufficient replications of a given file when $M=cte$,  we assume that the size of library $K=n^{1-\delta}$, where $0<\delta<1$ is a constant and $K=n$ otherwise.}
\fi

	 %Note that the required  communication cost
%from any given node guaranties  logarithmic replications of the requested file and the effect of power of two choices.

%In the last scheme, we consider both maximum load and communication cost objectives simultaneously. We achieve this goal by choosing the two candidate servers while putting a constraint on their distance to the requesting node. The harder the constraint is the lower the communication cost will be. On the other hand  hardening the constraint will correlate the two choices, which in turn, increases maximum load. This shows that there is a fundamental trade-off between communication cost and load balancing. The surprising fact is that, in large memory regime, without sacrificing much communication cost (i.e. $\Theta\left(n^\epsilon\right)$) one can achieve the benefits of power of two choices to reduce maximum load to $\Theta(\log\log n)$. We also demonstrate this trade-off by extensive simulations. 

\subsection{Related Work}
Load balancing has been the focus of many papers on cache networks \cite{CDN_Peng_2004, RoussopoulosB06, LeconteLM12}, among which distributed approaches have attracted a lot of attention (\eg, see \cite{ManfrediOR13}, \cite{AdlerCMR98}, and \cite{XiaAYL15}). Randomized load balancing via the power of two choices, is a popular approach in this direction \cite{Mitzenmacher01}. Chen et al. \cite{ChenLPCCSA05} consider the two choices selection process, where the second choice is the next neighbor of the first choice.  In \cite{XiaDH07} Xia et al. use the length of common prefix (LCP)-based replication to arrive at a recursive balls and bins problem. In \cite{ChenLPCCSA05} and \cite{XiaDH07}, the authors benefit from the metaphor of power of two choices to design algorithms for randomized load balancing. In contrast, in this paper we follow a theoretical approach to derive provable results for cache networks with limited memory. 

In \cite{ShahV15} the authors consider the supermarket model for performance evaluation of CDNs. Although the work \cite{ShahV15} considers the memory limitation into account, it does not consider the proximity principle which is a central issue in our paper. 
%Moreover, their analysis is limited to the scenarios with larger number of files than servers. This is not consistent with many practical scenarios where just a limited number of files are popular. 
Liu et al. \cite{LiuST16} study the setting where the clients compare the servers in terms of hit-rate (for web applications), or bit-rate (for video applications) to choose their favourite ones. Their setup and objectives are different from those we consider here. Moreover, they have not considered the effect of their randomized load balancing scheme on communication cost.

Additionally, the trade-off between proximity and load balancing in request routing has been considered in some works such as \cite{PathanVB08, TangTW14}, and \cite{StanojevicS09}. Although these works have mentioned this trade-off, non of them provides a rigorous analysis. To the best of our knowledge, our paper is the first work characterizing the above trade-off in an analytical framework.

From the theoretical viewpoint, in the standard balls and bins model, each ball (request) picks two bins (servers) independently and uniformly at random and it is then allocated to the one with lesser load \cite{ABKU99}. However, memory limitation and proximity principle in cache networks makes the bins choices correlated which  resembles the balls and bins model with \emph{related choices} (\eg, see \cite{BBFN12}, \cite{KP06}, \cite{God08}, and \cite{Pou16}). Our result also resides in this category, which is specific to cache networks with memory limitation and proximity constraint.

The organization of the paper is as follows. In Section~\ref{sec:Notation_ProblemSetting}, we present our notation and problem setup. Then, in Section \ref{sec:Nearest_Replica_Strategy} the \emph{nearest replica strategy}, as the baseline scheme, is presented and its maximum load and communication cost are investigated. In Section \ref{sec:Proximity-Aware_Two_Choices_Strategy}, we propose the \emph{proximity-aware two choices strategy}, which at the same time considers proximity of requests and servers, and benefits from the power of two choices. In order to do this, we first present some examples to shed light on different aspects of the problem. Then, we propose our main results in two different regimes, namely $M=n^{\alpha}$, for every constant $0<\alpha<1/2$, and $M=K$. In Section \ref{sec:Simulations} performance of these two schemes are investigated via extensive simulations. Finally, our discussions and future directions are presented in  \S\ref{sec:disscuss}.

%% file: 3-NotationProbSetting.tex
\section{Notation and Problem Setting}\label{sec:Notation_ProblemSetting}

\subsection{Notation}
Throughout the paper, with high probability refers to an event that happens with probability $1-1/n^c$, for some constant $c>0$. Let $G=(V,E)$ be a graph with vertex set $V$ and edge set $E$ where $e(G):= |E|$. For $u\in V$ let $d(u)$ denote for the degree of $u$ in $G$. For every pair of nodes $u,v\in V$,   $d_G(u, v)$ denotes the length of a shortest path from $u$ to $v$ in $G$. The neighborhood of $u$ at distance $r$ is defined as    
\[
  B_r(u) := \left\{v : d_G(u, v)\le r ~ {\text {and}}~~ v\in V(G) \right\}.
\]
Finally, we use $\mathrm{Po}(\lambda)$ to denote for the Poisson distribution with parameter $\lambda$.
%For convenience, we use $[m:n]$ to denote for $\{m,m+1,\ldots,n\}$ where $m$ and $n$ are arbitrary integer numbers.

\subsection{Problem Setting}

\iffalse
\begin{figure}
\begin{center}
\includegraphics[width=0.5\textwidth]{../Figures/Lattice.eps}
\end{center}
\caption{Network Model. \label{Fig_Ex_1}}
\end{figure}
\fi

We consider a cache network consisting of $n$ caching servers (also called cache-enabled nodes) and edges connecting neighboring servers forming a $\sqrt{n} \times \sqrt{n}$ grid. Thus, direct communication is possible only between adjacent nodes, and other communications should be carried out in a multi-hop fashion.

\begin{remark}
	Throughout the paper for the sake of presentation clarity we may consider a torus with $n$. This helps to avoid boundary effects of grid %which makes presentation more clear, 
	and all the asymptotic results hold for the grid as well.
\end{remark}

Suppose that the cache network is responsible for handling a library of $K$ files $\mathcal{W}=\{W_1,\dots,W_K\}$, whereas the popularity profile follows a known distribution $\mathcal{P}=\{p_1,\dots,p_K\}$.

The network operates in two phases, namely, \emph{cache content placement} and \emph{content delivery}. In the cache content placement phase each node caches $M\le K$ files randomly from the library according to their popularity distribution $\mathcal{P}=\{p_1,\dots,p_K\}$ with replacement, independent of other nodes.
Also note that, throughout the paper we assume that $M\ll K$, unless otherwise stated.

%\begin{remark}
%It should be noted that the cache content placement at each server can be implemented via efficient Distributed Hash Tables (DHT) schemes (see e.g., \cite{BauerHW07} and \cite{KargerLLPLL97}), which can adopt to dynamic library popularity profiles. This will also enable all users to obtain global cache content information  in a robust and distributed manner. In this paper we assume a static profile and do not go into the details of such schemes.
%\end{remark}

Consider a time block during which $n$ files are requested from the servers sequentially. The server of each request is chosen uniformly at random  from $n$ servers. Let $D_i$ denote the number of requests (demands) arrived at server $i$. Then for large $n$ we have $D_i \sim \mathrm{Po(1)}$ for all $1\le i\le n$.
%\alic{is it OK? to use $i$ for servers }

%In the content delivery phase requests of users getting service from this cache network arrive based on a Poisson process. Also we assume that the users getting service from this network are uniformly distributed so that on average each cache node is responsible for equal number of users. Consider a fixed time block such that one request on average should be served by each node in this time block. Thus, if we denote the number of requests arrived at node $i$ by $X_i$, then $X_i$'s are independently drawn from $\mathrm{Po}(1)$. Furthermore, we assume an ordering among all the requests based on their arrival times.

For library popularity profile $\mc{P}$, we consider two probability distributions, namely, Uniform and Zipf with parameter $\gamma$. In the  Uniform distribution we have
\begin{equation*}
p_i=\frac{1}{K},\quad i=1,\dots,K,
\end{equation*}
which considers equal popularity for all the files. In Zipf distribution the request probability of the $i$-th popular file is inversely proportional to its rank as follows
\begin{equation*}
p_i=\frac{1/i^\gamma}{\sum\limits_{j=1}^{K}1/j^\gamma},\quad i=1,\dots,K,
\end{equation*}
which has been confirmed to be the case in many practical applications \cite{Zipf1_99,Zipf2_07}.

\iffalse

Let us define %$\Lambda(\beta)\triangleq \sum_{j=1}^{K}1/j^\beta}$, then we have the following known result (e.g. see \cite{}).
\begin{fact}\label{prelm:1} For every given $\gamma>0$, we have 
\begin{equation*}
\Lambda(\gamma)=\left\{
\begin{array}{ll}
\Theta\left(K^{1-\gamma}\right),  & \quad 0 < \gamma <1, \\
\Theta\left(\log K\right) &\quad \gamma=1, \\
\Theta(1) &\quad \gamma>2.
\end{array}
\right.	
\end{equation*} 
\end{fact}

As discussed in the previous section, this cache network model can be applied to both wireless and wired scenarios. For example consider the network model in \cite{GolrezaeiHelper} which consists of distributed cache-enabled \emph{helpers} which deliver requested files to mobile wireless users in a single cell. These helpers may be connected  by a wired or a wireless backhaul. As another instance, the cache network can represents a Content Delivery Network (CDN) data centers which is geographically distributed, and is responsible for delivering content to users.
\fi

For any given cache content placement, an assignment strategy determines how each request is mapped to a server. 
Let $T_i$ denote the number of requests assigned to server $i$ at the end of mapping process.

Now, for each strategy we define the following metrics.
\begin{definition}[Communication Cost and Maximum Load]
~
\begin{itemize}
\item The \emph{communication cost} of a strategy is the average number of hops between the requesting node and the serving node, denoted by $C$. 
\item The \emph{maximum load} of a strategy is the maximum number of requests assigned to a single node, denoted by $L=\max_{1\le i\le n} T_i$. 
\end{itemize}
\end{definition}

%\pooyac{discussing cache contents knowledge DHT, and topology discovery}

%% file: 4-MinCommCostStrategy.tex
\section{Nearest Replica Strategy}\label{sec:Nearest_Replica_Strategy}
The simplest strategy for assigning requests to servers is to allocate each request to the nearest node that has cached the file. This strategy, formally defined below, leads to the minimum communication cost, while does not try to reduce maximum load.

\begin{definition} [Strategy I: Nearest Replica Strategy]
In this strategy each request is assigned to the nearest node --in the sense of the graph shortest path distance-- which has cached the requested file. If there are multiple choices ties are broken randomly.
\end{definition}

Consider the set of nodes that have cached file $W_j$, say $S_j$. According to Strategy I, each demand from node $u$ for file $W_j$ will be served by $\arg\min_{v \in S_j} d_G(u,v)$. This induces a Voronoi Tessellation on the torus corresponding to file $W_j$ which we denote by $\mathcal{V}_j$. Then, alternatively, we can define Strategy I as assigning each request of file $W_j$ to the corresponding Voronoi cell center. 

In order to analyze the maximum load imposed on each node, we should investigate the size of such Voronoi regions. The following Lemma is in this direction.
\begin{lemma}\label{WindMillLemma}
Under the Uniform popularity distribution,  the maximum cell size (number of nodes inside each cell) of $\mathcal{V}_j$, $1\le j\le K$, is at most $O\left(K \log n / M\right)$ w.h.p. In particular, every Voronoi cell centered at any node is contained in a sub-grid of size $r\times r$  with $r=O\left(\sqrt{K \log n / M}\right)$.
	 Furthermore, if $K=n^{1-\epsilon}$, for some  constant $0<\epsilon<1$, and $M=\Theta(1)$, then there exists a Voronoi cell of size $\Theta\left(K \log n / M\right)$ w.h.p.
	\end{lemma}

\begin{proof}
Refer to Appendix~\ref{apndx:proofs}.
\end{proof}

\iffalse	\begin{figure}[t]
		\vspace{-1cm}
		\centering
		\includegraphics[width=1.01\textwidth]{Figures/presentation1}
		\caption{\small{\small CAPTION.}\label{Upper} }
	\end{figure}
\fi

Now, we are ready to present our main results for this section which characterize the maximum load of Strategy I, in Theorems~\ref{thm:imblance}~and~\ref{thm:Strategy_I_Large}.

\begin{thm}\label{thm:imblance}
Suppose that $K=n^{1-\epsilon}$, for some constant $0< \epsilon <1$, and $M=\Theta(1)$. Then, under Uniform distribution $\mathcal{P}$, Strategy I achieves maximum load of $L=\Theta(\log n)$ w.h.p.
\end{thm}

\begin{proof}
	Consider node $u$ which has cached a set of distinct files, say $S$, with $|S|\le M$. Applying Lemma \ref{WindMillLemma} shows that all  Voronoi cells centered at $u$ corresponding to cached files at $u$ are contained in a sub-grid of size  at most $O(K \log n /M)$  w.h.p. Also in each round, every arbitrary node  requests for a file in $S$ with probability  $|S|/nK\le M/nK$, as each request  randomly chooses its origin and type. Hence, by union bound, a node in the sub-grid may request for a file in $S$  with probability at most $O(K\log n/M)\cdot (M/nK)=O(\log n/n)$. Since there are $n$ requests, the expected number of requests  imposed to node $u$ is $O(\log n)$. Now using a Chernoff bound (e.g., see Appendix \ref{app:bounds}) shows that w.h.p. $u$ has to handle   at most $O(\log n)$ requests.

\iffalse
 Number of file requests of each node follows $\mathrm{Po}(1)$, independent of other users (\ie,  $D_i \sim \mathrm{Po}(1)$ for $i=1,\dots,n$). Thus each node $v$ requests a specific file (say $W_j$) $Z_{v,j}$ times where $Z_{v,j} \sim \mathrm{Po}(1/K)$. 
Consider node $u$ which has cached file $W_j$. Then the load of type $j$ imposed on this node is 
\begin{equation}
L_j(u)=\sum_v Z_{v,j},
\end{equation}
where the summation is over  nodes in the Voronoi cell centered at $u$ corresponding to file $W_j$. Applying Lemma \ref{WindMillLemma} shows that each such Voronoi cell has at most $O(K \log n /M)$ nodes w.h.p. Thus $L_j(u)$ is the sum of at most $O(K \log n /M)$ independent $\mathrm{Po}(1/K)$ random variables. Thus $\Ex{L_j(u)} =O(\log n /M)$ and by applying a tai Chernoff inequality (\eg, see Appendix \ref{app:bounds}, Theorem~\ref{app:pois}), one can  easily see that $L_j(u)=O(\log n /M)$ w.h.p. Now since this node has cached $M$ files, the load of this node is at most $O(\log n)$.
\fi
On the other hand, to establish a lower bound on the maximum load we proceed as follows. Lemma \ref{WindMillLemma} shows that there exits a Voronoi cell in $\mathcal{V}_j$, for some $j$, such that the center node should handle the requests of at least $\Theta(K \log n /M)$ nodes w.h.p. Also each node in the cell may  request for file $W_j$ with probability $1/nK$.  So on average there are $\Theta(\log n/M)$ requests imposed on the cell center. Similarly, by  a Chernoff bound, one can see that this node experiences the load  $\Theta(\log n /M)$, which concludes the proof for constant $M$.
\end{proof}

\begin{remark}
It should be noted that the same result of $\Theta(\log n)$ for the maximum load can also be proved for the Zipf distribution. That is because the content placement distribution is chosen proportional to the file popularity distribution $\mathcal{P}$, and consequently this result is insensitive to $\mathcal{P}$. However, the proof involves lengthy technical discussions which we omit in this paper.
\end{remark}

\begin{thm}\label{thm:Strategy_I_Large}
	Suppose that $K=n$ and $M=n^{\alpha}$, for some $0< \alpha<1/2$. Then, under the Uniform distribution, the maximum load is in the interval $[\Omega(\log n/\log\log n), O(\log n)]$ w.h.p. %Moreover the average communication cost is $\sqrt{K/M}$
\end{thm}
\begin{proof}
	Refer to Appendix~\ref{apndx:proofs}.
\end{proof}

Next, we investigate the communication cost of Strategy~I in the following  theorem.

\begin{thm}\label{thm:Strategy_I_CommCost}
	Under the Uniform popularity distribution, Strategy I achieves the communication cost $C=\Theta(\sqrt{K/M})$, for every $M\ll K$. Under Zipf popularity distribution with $M=\Theta(1)$, it achieves
	\begin{equation}
		C = \left\{
		\begin{array}{llll}
		\Theta\left(\sqrt{K/M}\right) &: &\quad 0 < \gamma <1, \\
		\Theta\left(\sqrt{K/M \log K}\right) &: &\quad  \gamma=1, \\
		\Theta\left(K^{1-\gamma/2}/\sqrt{M}\right) &: &\quad  1< \gamma <2, \\ 
		\Theta\left(\log K /\sqrt{M}\right) &: &\quad \gamma=2, \\
		\Theta\left(1/\sqrt{M}\right) &: &\quad \gamma>2.
		\end{array}
		\right.	
		\end{equation} 
\end{thm}
\begin{proof}
Refer to Appendix~\ref{apndx:proofs}.
\end{proof}

Theorem \ref{thm:Strategy_I_CommCost}  shows how non-uniform file popularity reduces communication cost. The skew in file popularity is determined by the parameter $\gamma$ which will affect the communication cost. For example, for $\gamma<1$ communication cost is similar to the Uniform distribution, while for $\gamma > 2$, it becomes independent of $K$. 

Since in Strategy I we have assigned each request to the nearest replica, Theorem~\ref{thm:Strategy_I_CommCost} characterizes  the minimum communication cost one can achieve. 
However, Theorems~\ref{thm:imblance}~and~\ref{thm:Strategy_I_Large} show a logarithmic growth for the maximum load as a function of network size $n$. %which may not be acceptable for large networks. 
This imbalance in the network load is because in Strategy I each request assignment does not consider the current load of servers. A natural question is whether, at each request allocation, one can use a very limited information of servers' current load in order to reduce the maximum load. Also one can ask how does this affect the communication cost.

%% file: 5-PowTwoChoiceStrategy.tex
\section{Proximity-Aware Two Choices Strategy}\label{sec:Proximity-Aware_Two_Choices_Strategy}
Strategy I introduced in the last section will result in the minimum communication cost, while, the maximum load for that strategy is of order $\Omega\left(\log n/\log\log n\right)$. In this section we investigate an strategy which will result in an exponential decrease in the maximum load, \ie, reduces maximum load to $\Theta\left(\log \log n \right)$, formally defined as follows.

\begin{definition}[Proximity-Aware Two Choices Strategy]
For each request born at an arbitrary node $u$ consider two uniformly at random chosen nodes from $B_r(u)$, that have cached the requested file. Then, the request is assigned to the node with lesser load. Ties are broken randomly.
\end{definition}

For the sake of illustration, first, we consider some examples in the following.

\begin{example}[$M=K$ and $r=\infty$\footnote{It should be noted that $r\ge \sqrt{n}$ (including $r=\infty$) is equivalent to $r= \sqrt{n}$. Thus in this paper we use $r=\sqrt{n}$ and $r=\infty$ alternatively.}]\label{ex:PO2C_M=K_r=infty}
In this example each node can store all the library and there is no constraint on proximity. As mentioned in \S~\ref{sec:Introduction}, the number of files that should be handled by each node (\ie,  $D_i$ for $i=1,\dots,n$) will be a $\mathrm{Po}(1)$ random variable.  In this case, according to Strategy~II, two random nodes are chosen from all network nodes and the request is assigned to the node with lesser load.

Therefore, in terms of maximum load, this  problem is reduced to the standard power of two choices  model in the balanced allocations literature \cite{ABKU99}. In this model there are $n$ bins and $n$ sequential balls which are randomly allocated to bins. In every round each ball picks two random bins uniformly, and it is then allocated to the bin with lesser load \cite{ABKU99}. Then it is shown that the maximum load of network is $L=\max_i T_i=\log \log n(1+o(1))$ w.h.p. \cite{ABKU99}, which is an exponential improvement compared to Strategy I. %Moreover, the communication cost of this example will be $\Theta(\sqrt{n})$.
\end{example}
  
However, in contrast to Example~\ref{ex:PO2C_M=K_r=infty}, in cache networks usually each node can store only a subset of files, and this makes the problem different from the standard balls and bins model, considered in \cite{ABKU99}. Here, due to the memory constraint at each node, the choices are much more limited than the $M=K$ case. In other words here we have the case of \emph{related choices}.
In the related choices scenario, the event of choosing the second choice is correlated with the first choice; this correlation may annihilate the effect of power of two choices as demonstrated in Example~\ref{ex:PO2C_K=n_M=cte_r=infty}.

\begin{example}[$K=n$, $M=\Theta(1)$, and $r=\infty$]\label{ex:PO2C_K=n_M=cte_r=infty}
In this regime,  there is a subset of the library, say $S$ with $|S|=\Theta(n)$, whose files are  replicated  in $[1,M]$ number of places. On the other hand, each file type is requested $\mathrm{Po}(1)$ times and hence w.h.p. there will be a file in $S$ which is requested $\Theta(\log n/\log\log n)$ times (e.g., see \cite{Devroye85}). Since each file in $S$ is replicated at most $M$ times, requests for the file are distributed among at most $M$ nodes and thus
  the maximum load of the corresponding nodes will be at least $\Theta(\log n/\log\log n)/M$. Hence, due to memory limitation we cannot benefit from the power of two choices.
\end{example}

Although Example~\ref{ex:PO2C_K=n_M=cte_r=infty} shows that memory limitation can annihilate the power of two choices but this is not always the case. Example~\ref{ex:prop_M1} shows even for $M=1$ for some scenarios we can achieve $L=O(\log\log n)$.

%The simplest scenario which can be considered in this direction is the case $M=1$. In this case we show that by guaranteeing $K = n^{1-\epsilon}$ for any constant $0<\epsilon<1$ we can still benefit from power of two choices, as proved in the following proposition.
 
\begin{example}[$K=n^{1-\epsilon}$ for every constant $0<\epsilon<1$, $M=1$, and $r=\infty$] \label{ex:prop_M1}
For any popularity distribution $\mc{P}$ where $\sum_{j=1}^{K}{(p_jn)^{-c}}=o(1)$, Strategy~II achieves maximum load $L=O(\log \log n)$ w.h.p. Also, notice that Uniform and Zipf distributions   satisfy this requirement,
whenever $\epsilon\in \left( \frac{\gamma-1}{\gamma}, 1 \right)$ for $\gamma\ge 1$, where $\gamma$ is  Zipf parameter.

Roughly speaking, when $M=1$, we may partition the servers based on their cached file and hence we have $K$ ``disjoint'' subsets of servers. Similarly there are $K$ request types where each request should be addressed by the corresponding subset of servers. Thus, here we have $K$ disjoint Balls and Bins sub-problems, and the sub-problem with maximum load determines the maximum load of the original setup. 
The reason that here, in contrast to Example~\ref{ex:PO2C_K=n_M=cte_r=infty}, we can benefit from power of two choices is the assumption of $K\ll n$.

For a formal proof of above claim, refer to Appendix \ref{apndx:Examp-M1}.
\end{example}

Above examples bring to attention the following question.

\begin{question}\label{que:MemoryLimitRegimes}
In view of the memory limitation at each server in cache networks, what are the regimes (in terms of problem parameters) one can benefit from the power of two choices to balance out the load?
%A question that we address in this section is that --despite this memory limitation-- can we still benefit from the power of two choices? 
\end{question}

Addressing Question~\ref{que:MemoryLimitRegimes}, for the general $M>1$ case, is more challenging than Example~\ref{ex:prop_M1} and needs a completely different approach.
The simplicity of case $M=1$ is that there is no interaction between $K$ Balls and Bins sub-problems. On the other hand, consider $M>1$. If a request, say $W_j$, should be allocated to a server then the load of two candidate bins that have cached $W_j$ should be compared. However, load of other file types will also be accounted for in this comparison. So there is flow of load information between different sub-problems which makes them entangled. 

In all above examples, we have not considered the proximity constraint, \ie, $r=\infty$, yet. This results in a fairly high communication cost $C=\Theta\left(\sqrt{n}\right)$. However, in general since parameter $r$ controls the communication cost, it can be chosen to be much less than the network diameter, \ie, $\Theta(\sqrt{n})$. This proximity awareness introduces another source of correlation (other than memory limitation) between the two choices. Thus, considering the proximity constraint may annihilate the power of two choices even in large memory cases as demonstrated in the following example.

%Until now we focused in the case $r=\infty$  reduced maximum load by considering all the nodes which could satisfy the request, regardless of their distance to the requester.  This results in fairly high communication cost $C=\Theta\left(\sqrt{n}\right)$. However, as defined in Definition ????, in order to reduce communication cost we should limit $r$ \pooyac{Discussion about the role of correlation  introduced by diatnace.}

\begin{example}[$M=K$ and $r=1$]
In this example, when a request arrives at a server, the server chooses two random choices among itself and its neighbours. Then the request is allocated to the one with lesser load. Since there exists a server at which $\max_i D_i=\Theta(\log n/\log\log n)$ requests arrive, maximum load of network (\ie, $L=\max_i T_i$) will be at least $\Theta(\log n/\log\log n)/5$.
\end{example}

Thus, similar to Question~\ref{que:MemoryLimitRegimes} regarding the memory limitation effect, one can pose the following question regarding proximity principle.

\begin{question}\label{que:ProximityConstraint}
In view of the proximity constraint of Scheme~II, what are the regimes (in terms of problem parameters) one can benefit from the power of two choices to balance out the load?
\end{question}

In order to completely analyze load balancing performance of Scheme~II, one should consider both sources of correlation simultaneously (which is not the case in above examples). To this end, in the following, we investigate two memory regimes, namely $M=K$ and $M=n^{\alpha}$, for some $0<\alpha<1/2$.
\input{Proof-of-Proximity-Aware-Regime-I}

\input{Proof-of-Proximity-Aware-Regime-II}

%% file: Proof-of-Proximity-Aware-Regime-I.tex
Our main result for $M=n^{\alpha}$ is presented in the following theorem.

%Before we state our main result, let us assume that 
% $r=n^{\alpha}$ and $M=n^{\beta}$, where  $1/2-\alpha$ and $1/2-\beta$ are bounded away from zero, respectively.

\begin{thm}\label{main:twochoice}
%Suppose that $n$ caching servers are organized as a torus.  
Suppose that $0< \alpha, \beta<1/2$ be two constants and let $K=n$, $M=n^{\alpha}$, and $r=n^{\beta}$. Then if 
\[
\alpha+2\beta\ge 1+2\log\log n /\log n,
\]
under the Uniform popularity distribution,  Strategy II achieves maximum load $L=\Theta(\log \log n)$ and communication cost $C=\Theta(r)$ w.h.p.
\end{thm}

\begin{remark}
To have a more accessible proof, in Theorem~\ref{main:twochoice}, we have assumed that $K=n$.
%, and $\alpha$ and $\beta$ are positive constants. 
Note that the proof techniques can also be extended to the case where $K=O(n)$.
\end{remark}

%The rest of this subsection is devoted to the proof of Theorem~\ref{main:twochoice}.  
In order to prove the theorem, let us first present an interesting result that was shown in \cite{KP06} as follows.
\begin{thm}[\cite{KP06}]\label{thm:related}
	Given an almost $\Delta$-regular graph\footnote{A graph is said to be almost $\Delta$-regular, if each vertex has degree $\Theta(\Delta)$.} $G$ with $e(G)$ edges and $n$ nodes
	representing $n$ bins, if $n$ balls are thrown into the bins
	by choosing a random edge with probability at most $O(1/e(G))$   and placing into the smaller
	of the two bins connected by the edge, then the maximum
	load is $\Theta(\log \log n) + O\left( \frac{\log n}{\log(\Delta/\log^4 n )} \right) + O(1)$ w.h.p.
\end{thm}
\begin{remark}
Note that in the original theorem presented in \cite{KP06}, it is assumed that each edge is chosen uniformly among all edges of graph $G$. However, here we slightly generalize the result so that each edge is chosen with probability at most $O(1/e(G))$. The proof follows the original proof's idea with some  modifications in computation parts, where due to lack of space we omit.
\end{remark}
In order to apply Theorem~\ref{thm:related}, we first need to define a new graph $H$ as follows.
\begin{definition}[Configuration Graph]
For  given parameter $r$, configuration graph $H$ is defined as a graph	whose nodes represent  the servers and two nodes, say $u$ and $v$, are connected if and only if  $u$ and $v$ have cached a common file and $d(u, v)\le 2r$ in the torus.
\end{definition}

%Recall that   the cache placement strategy proceeds in $M$ steps, where in each of which every  server independently caches a file according to file's popularity. 

For every two servers $u$ and $v$,
let $T(u, v)$ be the set of distinct files that have been cached in both nodes $u$ and $v$. Also denote $|T(u, v)|$ by $t(u, v)$. Define $t(u)$ to be the number of distinct  cached files  in $u$. Now, let us define \emph{goodness} of a placement strategy as follows.
\begin{definition}[Goodness Property]
For every positive constant $\delta\in[0, 1]$ and $\mu=O(1)$, we say the  file placement strategy is \emph{$(\delta, \mu)$-good}, if for every $u$ and $v$, $t(u)\ge \delta M$ and $t(u, v)<\mu$.
\end{definition}

\begin{lemma}\label{lem:GoodnessProperty}
The proportional cache placement strategy introduced in \S\ref{sec:Notation_ProblemSetting}, is $(\delta, \mu)$-good w.h.p.
\end{lemma}
\begin{proof}
	Clearly, every set of cached files in every node (with replacement) can be one-to-one mapped to a non-negative integral solution of equation $\sum_{i=1}^Kx_i=M$, where each $x_i$ expresses the number of times  that file $i$  has been cached in the node. A combinatorial argument shows that, the equation has ${K+M-1 \choose M}$ non-negative integer solutions. So for each $1\le s\le  M$, we have
	\begin{equation}\label{eq:Prob_c(u)=s}
	\Pr{t(u)=s}=\frac{{ K \choose s}{M-1 \choose M-s}}{{ K+M-1 \choose M}},
	\end{equation}
	where we first fixed a set of file indexes   of size $s$, say $I=\{i_1, i_2,\ldots, i_s\}$, and then count the number of integral solutions to the equation $\sum_{i\in I}x_i=M-s$.

In order to bound \eqref{eq:Prob_c(u)=s}, we note that for every $1\le a\le b,$ $(b/a)^a \le  {b \choose a} \le b^a$ and also ${b \choose a}\le 2^b$.  Recall that we assumed  $K=n$ and $M=n^\alpha$, $0<\alpha<1/2$. Hence for every $1\leq s \leq \delta M$, we have
	\begin{align*}
	\Pr{t(u)=s}&\le \frac{K^s2^{M}}{{K\choose M}}\le 
	\frac{K ^s2^{M}}{(K/M)^M}=
	{(2M)^M}{K^{s-M}}\\
	&
	\leq  (2n^\alpha n^{\delta-1})^M.
	\end{align*}
 Thus, by choosing $\delta=(1-\alpha)/3$, for every \mbox{$1\le s\le \delta M $}, we have
	\begin{align*}
	\Pr{t(u)=s}&\le (2n^{\alpha+\delta-1})^M=
	(2n^{2\alpha/3-2/3})^M\\
	& \le (2n^{-1/3})^M=n^{-\omega(1)},
	\end{align*}
	where the last equality follows due to $M=n^{\alpha}=\omega(1)$.
	Now the union bound over all $1\le s\le \delta M$ and $n$ nodes yields 
	\begin{align}\label{prob:size}
	\Pr{\exists u\in V: t(u) \le \delta M}=n^{-\omega(1)}.
	\end{align}
	By a similar argument, for each $1\le t\le M$ and every $u$ and $v$, we have
	\begin{align*}%\label{eq:interest}
	\Pr{t(u, v)\geq t}={K\choose t}\left(\frac{{K+M-t-1 \choose M-t}}{{K+M-1 \choose M}}\right)^2.
	\end{align*}
	Thus,  for any constant $\mu \ge 5/(1-2\alpha)$, we can write
	\begin{align*}
	&\Pr{t(u, v)\ge \mu}\\
	&\le K^\mu \left(\frac{(K+M-\mu-1)!M!}{(K+M-1)!(M-\mu)!}\right)^2\\
	&\le K^\mu \left(\frac{M^\mu}{K^\mu}\right)^2\le \frac{M^{2\mu}}{K^\mu}= n^{(2\alpha-1)\mu}=O(1/n^5).
	\end{align*}
 By applying the union bound over all pairs of servers, for every $u$ and $v$ we have 
	\begin{align}\label{prob:com}
	\Pr{t(u, v)\ge  \mu}=O(1/n^3).
	\end{align}
	Hence, $t(u, v) < \mu$ w.h.p.
Putting  inequalities (\ref{prob:size}) and (\ref{prob:com}) together  concludes the proof.
\end{proof}

The following lemma presents some useful properties of $H$ and Strategy II.
\begin{lemma}\label{lem:Gp}
	Conditioning on goodness of file placement and assuming $K=n$, $M=n^\alpha$ and $r=n^{\beta}$ with $\alpha+2\beta\ge 1+2\log\log n/\log n$, we have
	\begin{itemize}
		\item[(a)] W.h.p. $H$ is almost $\Delta$-regular with $\Delta=\Theta\left(\frac{M^2r^2}{K}\right)$.
		\item[(b)] For each request, Strategy II samples an edge of $H$ (two servers) with probability $O(1/e(H))$.
	\end{itemize}
\end{lemma}
\begin{proof}
	Consider arbitrary node $u$ with $s$ distinct files. Then by definition of $H$, for every node $v$  we have 
	\begin{align*}
	p_s:=\Pr{t(u, v)\ge 1| t(u)=s }&=1-\left(\frac{K-s}{K}\right)^M\\
	&=\frac{sM}{K}(1+o(1)),
	\end{align*}
	where $1\le s\le M$.
	On the other hand $u$ and $v$ are connected in $H$, if in addition $d_G(u, v)\leq 2r$.
	Therefore for every given node $u$ with $s$ distinct cached files, $d(u)$ in $H$ (degree of $u$ in $H$) has a binomial distribution $\mathrm{Bin}(b_{2r}(u), p_s)$, where  $b_{2r}(u)=|B_{2r}(u)|$. Hence applying a Chernoff bound implies that with probability $1-n^{-\omega(1)}$, we have 
	\[
	d(u)=\frac{sMb_{2r}(u)}{K}(1+o(1)).
	\] 

	Conditioning on the goodness of file placement, $s=t(u)=\Theta(M)$.
	Also  by symmetry of torus, we have $b_{2r}(u)=\Theta(r^2)$, for every $u$. So, with high probability for every $u$, we have
	\[
	d(u)=\Theta\left({M^2r^2}/{K}\right),
	\]
	where this concludes the proof of part (a).

	Now it remains to show that Strategy~II picks an edge of $H$, with probability $O(1/e(H))$. First, notice that
	\begin{align}\label{eq:edgesize} e(H)=\Theta\left({nM^2r^2}/{K}\right)=\Theta(M^2r^2),
	\end{align}
    as $K=n$. 
	Then recall that each file is cached in every node with probability 
	$p=1-(1-1/K)^M=M/K(1+o(1))$, independently. 
	For any given node $u$ and file $W_j$, let $F_j(u)$ be the number of nodes at distance at most $r$ that have cached file $W_j$. Then $F_j(u)$ has a binomial distribution $\mathrm{Bin}(b_r(u), p)$, where $b_r(u)=|B_r(u)|$. So 
	\[
	\Ex{F_j(u)}=b_r(u)\cdot p=\Theta(r^2M/K),
	\] 
	where $b_r(u)=\Theta(r^2)$ for every $u$.
	Since \mbox{$\alpha+2\beta\ge 1+2\log\log n/\log n$} we have $\Ex{F_j(u)}=\omega(\log n)$, for every  $u$ and $j$.
	Now, applying a Chernoff bound for $F_j(u)$ implies that with probability $1-n^{-\omega(1)}$, $F_j(u)$ concentrates around its mean and hence, w.h.p., we have for every $u$ and $j$
	 \[	
	F_j(u)=\Theta(r^2M/K)=\Theta(r^2M/n).
	\] 
	
Consider an edge $(u, v)\in E(H)$, with $t(u, v)=t$. Define $S_{u,v}$ to be  the set  of nodes that may pick pair $u$ and $v$ randomly in Strategy~II. It is not hard to see that $|S_{u, v}|=O(r^2)$. Now we have,	
	\begin{align}
	&\Pr{ (u, v)\in E(H) \text{ is picked by Strategy II}| t(u, v)=t} \nonumber\\
	&\quad\quad =\sum_{j\in T(u, v)}\frac{1}{K}\sum_{w\in S_{u, v}}\frac{1}{n}\frac{1}{{F_j(w)\choose 2}} \nonumber\\
	&\quad\quad =\frac{1}{n^2}
	\sum_{j\in T(u, v)}\sum_{w\in S_{u, v}}\frac{1}{{F_j(w)\choose 2}} \nonumber\\
	&\quad\quad =\frac{1}{n^2}
	\sum_{j\in T(u, v)}\sum_{w\in S_{u, v}}
	\Theta({n^2}/{r^4M^2}). \label{eq:CondPr_(u,v)_picked}
	\end{align}
Conditioned on ``goodness,'' we have  for every $(u, v) \in E(H)$, $1\le t(u, v)< \mu$. So  \eqref{eq:CondPr_(u,v)_picked} can be simplified as 
\begin{align*}
&\Pr{ (u, v)\in E(H) \text{ is picked by Strategy II}}\\
&\quad\quad\le \Theta({ \mu |S_{u, v}|}/{r^4M^2})\\
&\quad\quad =O(1/r^2M^2)=O(1/e(H)),
	\end{align*}
	where the last equality follows from \eqref{eq:edgesize}.
\end{proof}

\begin{proof}[Proof of Theorem \ref{main:twochoice}]
	Applying Lemma \ref{lem:Gp} shows that w.h.p. the configuration graph $H$ is an almost $\Delta$-regular graph where $\Delta=M^2r^2/n$. Moreover, in each step, every edge of $H$ is chosen randomly with probability $O(1/e(H))$. 
	Hence, we can apply Theorem~\ref{thm:related} and conclude that w.h.p. the maximum load is at most
	\[
\Theta(\log\log n)+ O\left( \frac{\log n}{\log(\Delta/\log^4 n )} \right) =\Theta(\log\log n)+ O(1),
	\]
where it follows because $\alpha+2\beta\ge 1+2\log\log n/\log n$ and hence $\Delta=M^2r^2/n=n^{2\alpha+2\beta-1}>n^{\alpha}$.
\end{proof}

%% file: Proof-of-Proximity-Aware-Regime-II.tex
Now let us present our next result regarding to the $M=K$ regime.

\begin{thm}\label{thm:full}
 Suppose $M=K$ and Uniform distribution $\mathcal{P}$ over the file library. Then Strategy II achieves the maximum load $L=\Theta\left(\log \log n\right)$ and communication cost $C=\Theta\left(n^{\beta}\right)$ for any  $\beta=\Omega(\log\log n/\log n)$.
 \end{thm}
 
 \begin{proof}
\iffalse Let us first define the graph power as follows.
 \begin{definition}
   For every integer $r\ge 1$, the \emph{$r$-th power of  $G$}, denoted by $G^r$, is a graph whose vertex set is the same as $G$ and nodes $u$ and $v$ are connected if and only if $v\in B_r(u)$.
   \end{definition}
\fi
 %Then, it is easy to see that for a torus with $n$ nodes, for any given node $u$, $|B_r(u)|=\Theta(r^2)$, so by taking
  Let us choose $r=n^{\beta}$, for some $\beta=\Omega(\log\log n/\log n)$. 
 % $G^r$ is a $\Delta$-regular graph with $\Delta=\Theta(n^\epsilon)$. 
By the assumption $M=K$, the configuration graph $H$ (corresponding to $r$) is a graph in which two nodes $u$ and $v$ are connected if and only if $d(u,v)\le 2r$. Since our network is symmetric, for every $u$, $|B_r(u)|=\Theta(r^2)$ and hence $H$ is a  regular graph with $\Delta=\Theta(r^2)$.  
Also it is not hard to see that  Strategy II is equivalent to choosing an edge uniformly from $H$. Applying Theorem \ref{thm:related} (\cite{KP06}) to $H$ results in the maximum load of $\Theta(\log\log(n))$. In addition, choosing two random nodes in $|B_r(u)|=\Theta(r^2)$ results in communication cost of $C=\Theta(r)=\Theta\left(n^{\beta}\right)$. 
 %
% Then, we say that a regular graph is dense-enough if its degree is at least $n^\epsilon$. It is easy to see that for a torus with $n$ nodes,
 % the graph is symmetric and
% for any given node $u$, $|B_r(u)|=\Theta(r^2)$, so by taking $r=n^{\epsilon/2}$, $G^r$ is dense-enough. 
%
%Thus Strategy II is equivalent to choosing an edge uniformly from $G^r$ and comparing 
 %
%  the second choice as one of the neighbors of $u$ in $G^r$, which will result in maximum load $L=\Theta(\log \log n)$, due to Theorem \ref{thm:related} for dense-enough graphs \cite{KP06}. In addition, choosing two random nodes in $|B_r(u)|=\Theta(r^2)$ results in communication cost of $C=\Theta(r)=\Theta\left(n^{\epsilon/2}\right)$. 
 \end{proof}
 
 The main point of Theorem \ref{thm:full} is that we can just have $C=\Theta\left(n^{\beta}\right)$, for $\beta=\Omega(\log\log n/\log n)$, to benefit from the luxury of power of two choices, which is a very encouraging result.

%% file: 7-Simulations.tex
\section{Simulations}\label{sec:Simulations}
In this section, we demonstrate the simulation results for two strategies introduced in the previous sections, namely, \emph{nearest replica} and \emph{proximity aware two choices} strategies. The simulation results are shown for the torus topology.
Here, we consider Uniform popularity over the file library. As a result, the file placement is also considered to be uniform over the servers' storage. 

%In the first set of simulations, Strategy~I is simulated over a torus with various number of servers and cache sizes. 
Figure~\ref{fig:OneChoice_SrvSzVar_MaxLoad} shows the maximum load of Strategy~I as a function of the number of servers where different curves correspond to different cache sizes. The network graph is a torus, where $100$ files with Uniform popularity are placed uniformly at random in each node. Each point is an average of $10000$ simulation runs.
This figure confirms that the logarithmic growth of the maximum load, asymptotically proved in Theorem~\ref{thm:imblance}, also holds for intermediate values of $n\approx 100,\ldots, 3000$ which makes the result of Theorem~\ref{thm:imblance} more general. 
Comparing different curves reveals the fact that in larger cache size setting, we have a more balanced network. That happens because enlarging cache sizes results in a more uniform Voronoi tessellation, \ie, having cells with smaller variation in size.

Furthermore, Figure~\ref{fig:OneChoice_CacheSzVar_AvgCost} shows the communication cost of Strategy~I as a function of cache size where different curves correspond to different library sizes. Here, the network graph is a torus of size $2025$ and each point is an average of $10000$ simulation runs. This figure is in agreement with the result of Theorem~\ref{thm:Strategy_I_CommCost}.

%Figures~\ref{fig:OneChoice_SrvSzVar_AvgCost}~and~\ref{fig:OneChoice_SrvSzVar_MaxLoad} show the average cost and maximum load of strategy~I when the number of servers vary from $2025$ to $50176$. To average out the randomness of each simulation instance, every point on these figures is an average of $800$ runs of independent simulations. 
%As it is inferred from Figure~\ref{fig:OneChoice_SrvSzVar_MaxLoad}, if the size of cache in each server, $M$, is larger than some value then the maximum load does not depend very much on $M$.   
%Figure~\ref{fig:OneChoice_SrvSzVar_AvgCost} confirms the result of Theorem~\ref{thm:Strategy_I_CommCost} for uniform popularity.
%\mahdic{What more can we say about the maximum load figure (Figure~\ref{fig:OneChoice_SrvSzVar_MaxLoad})?}

%\subsection{Minimum Communication Cost Strategy}
%The simulations are presented in Figure~\ref{fig:OneChoice_SrvSzVar_AvgCost} and Figure~\ref{fig:OneChoice_SrvSzVar_MaxLoad}.

\begin{figure}
\begin{center}
\includegraphics[width=0.48\textwidth]{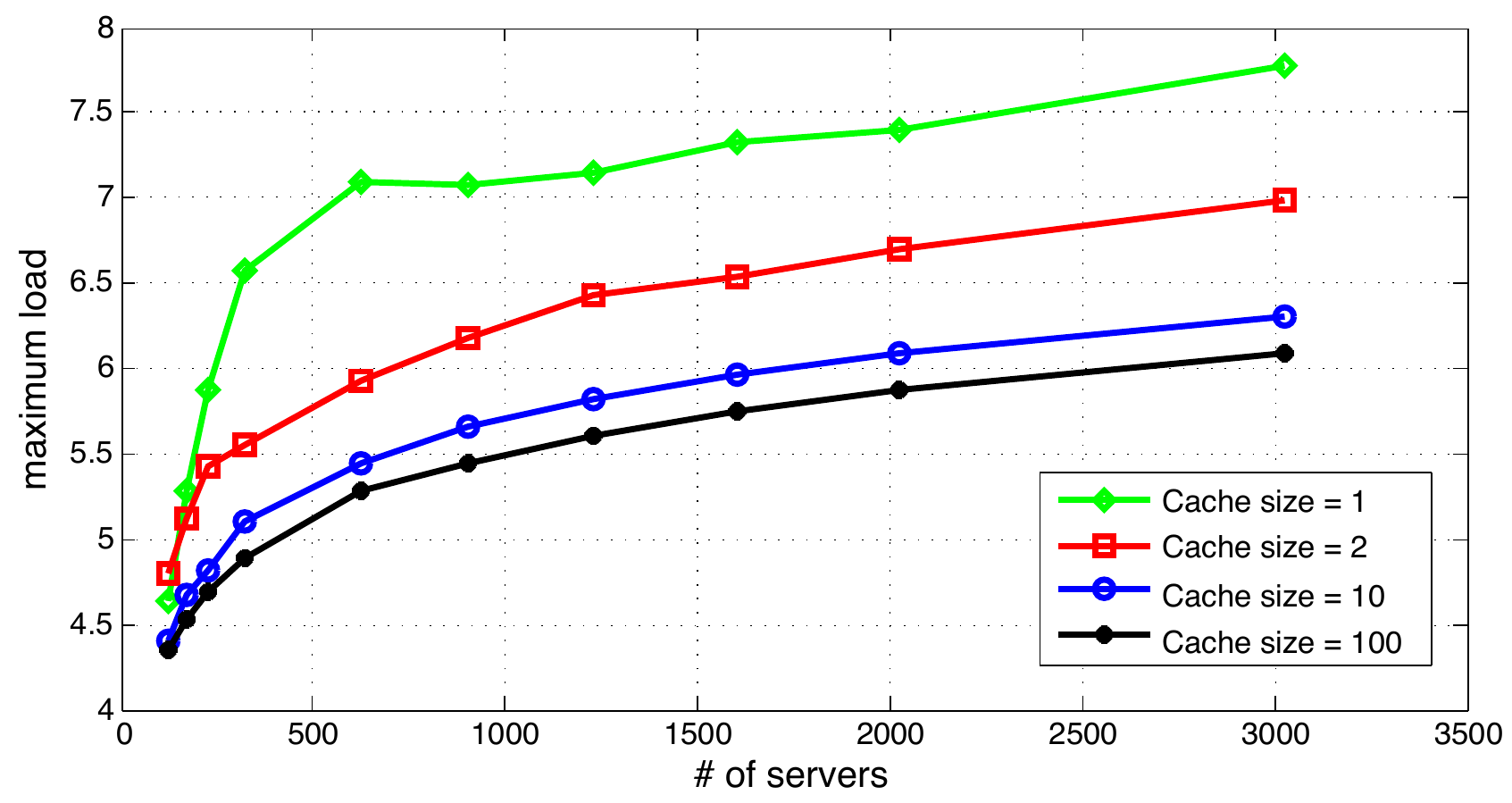}
\end{center}
\caption{The maximum load versus number of servers for Strategy~I.}\label{fig:OneChoice_SrvSzVar_MaxLoad}
\end{figure}

\begin{figure}
\begin{center}
\includegraphics[width=0.48\textwidth]{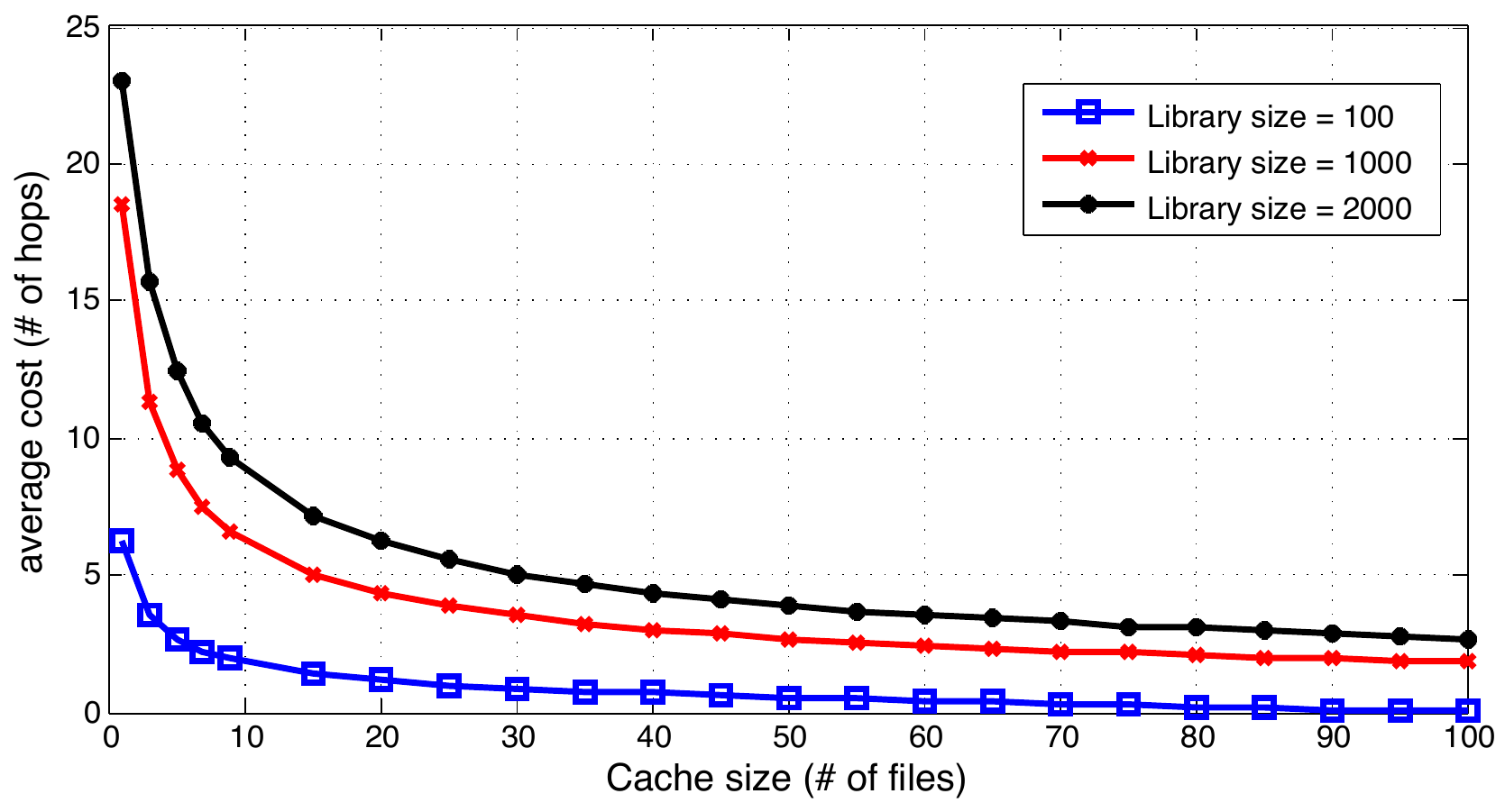}
\end{center}
%\caption{The average cost versus number of servers for Strategy~I. The network graph is a torus, where $100$ files with uniform popularity are placed uniformly at random in each node. Each point is an average of $10000$ simulation runs.}\label{fig:OneChoice_SrvSzVar_AvgCost}
\caption{The communication cost versus cache size for Strategy~I.}
\label{fig:OneChoice_CacheSzVar_AvgCost}
\end{figure}

\begin{figure}
\begin{center}
\includegraphics[width=0.48\textwidth]{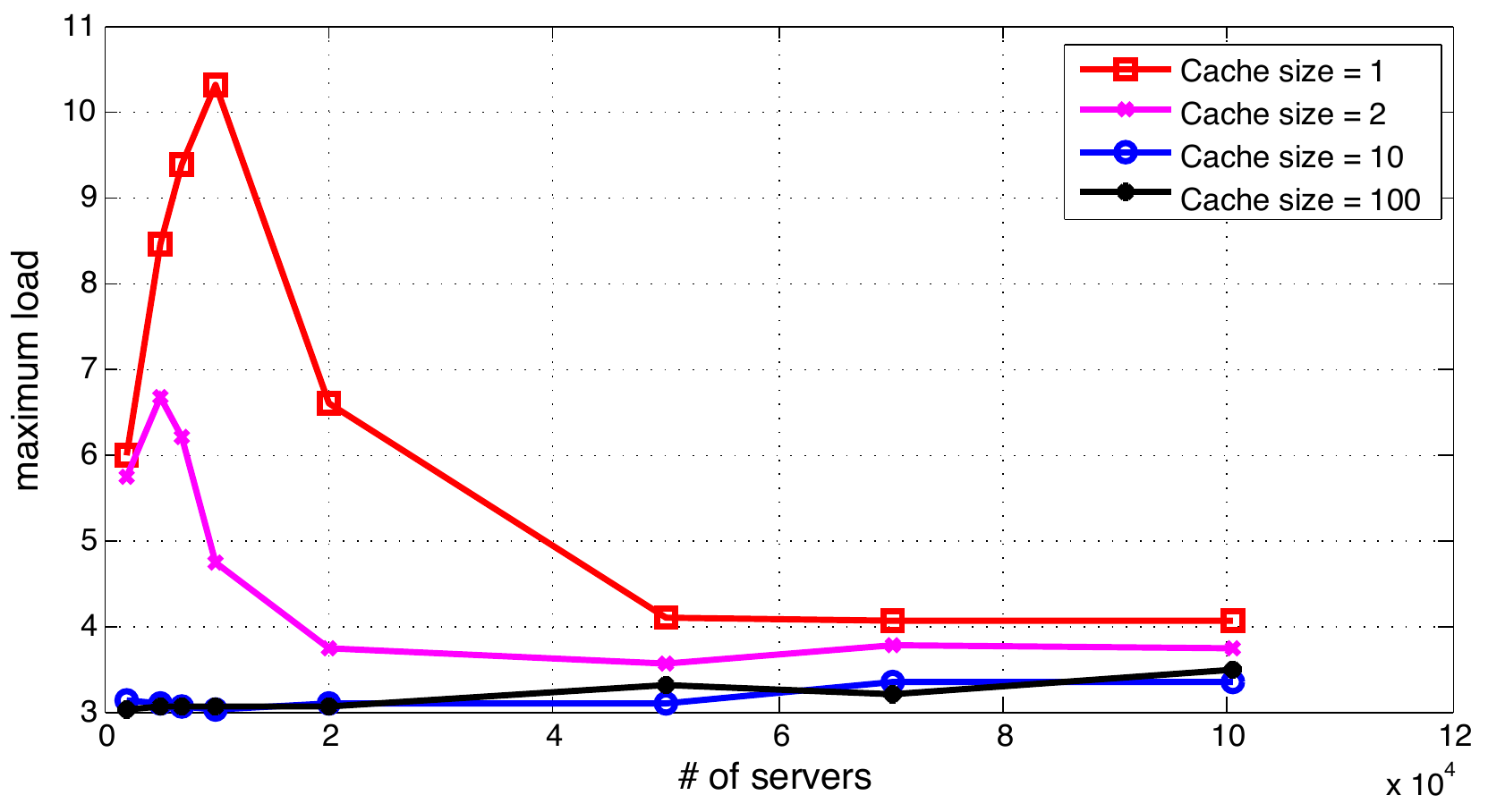}
\end{center}
\caption{The maximum load versus number of servers for Strategy~II. Here, we assume $r=\infty$.}\label{fig:TwoChoice_SrvSzVar_MaxLoad}
\end{figure}

\begin{figure}
\begin{center}
\includegraphics[width=0.48\textwidth]{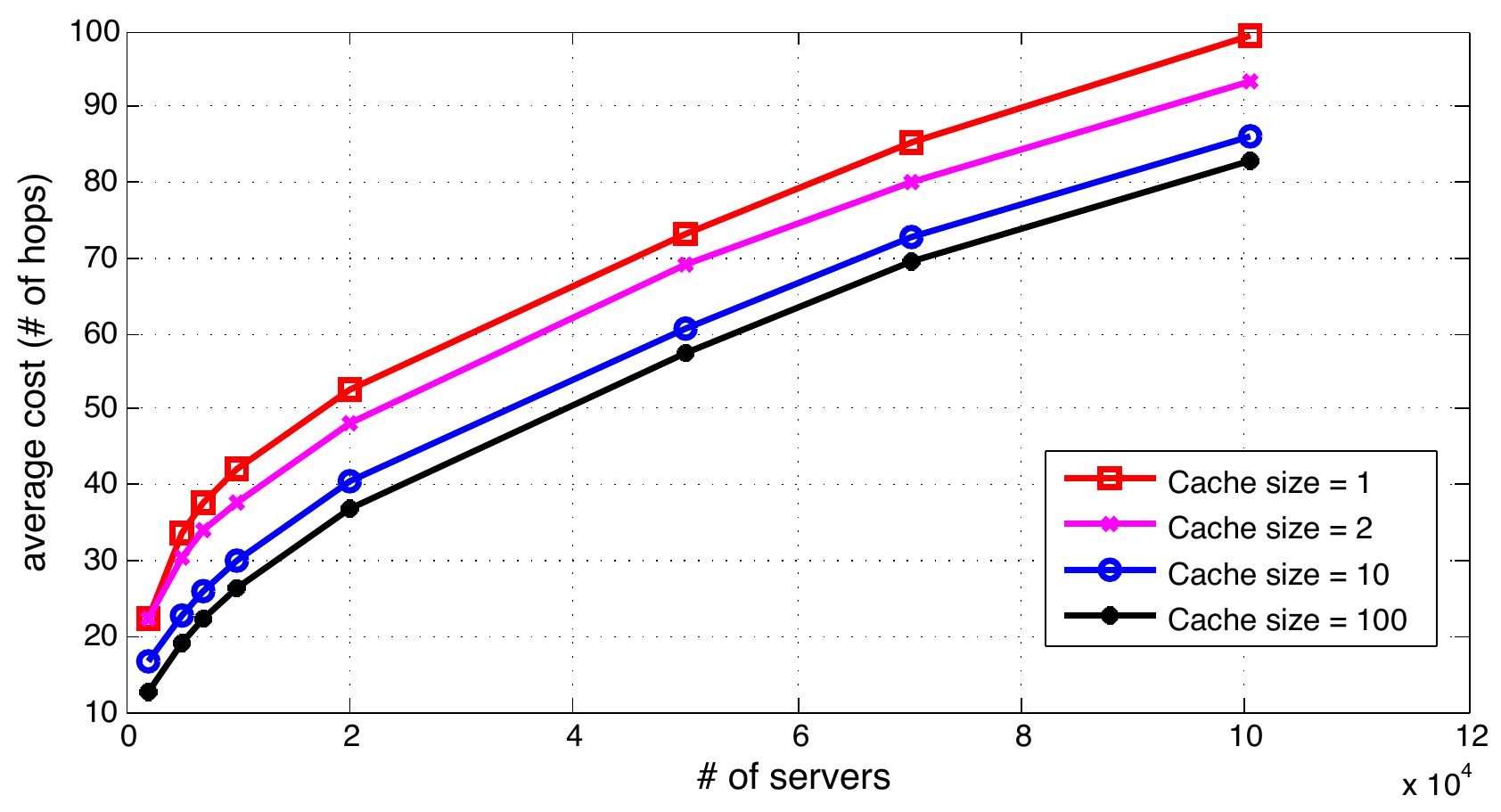}
\end{center}
\caption{The communication cost versus number of servers for Strategy~II. Here, we assume $r=\infty$.}\label{fig:TwoChoice_SrvSzVar_AvgCost}
\end{figure}

In order to simulate Strategy~II, first we set $r=\infty$ to study the effect of cache size on the maximum load and communication cost and then consider the effect of limited $r$ on the performance of the system. Figure~\ref{fig:TwoChoice_SrvSzVar_MaxLoad} shows the maximum load of the network versus number of servers where each curve demonstrates a different cache size. The network graph is a torus, where $2000$ files with Uniform popularity are placed uniformly at random in each node. Each point is an average of $800$ simulation runs.
In each curve, since cache size and number of files are fixed, increasing the number of servers translates to increasing each file replication. 

In Figure~\ref{fig:TwoChoice_SrvSzVar_MaxLoad}, when the file replication is low, due to high correlation between the two choices of Strategy~II, power of two choices is not expected. This is reflected in Figure~\ref{fig:TwoChoice_SrvSzVar_MaxLoad}; for example in the curve corresponding to $M=1$ for $n \le 10000$ we have a fast growth in maximum load which mimics the load balancing performance of Strategy~I. However, for $n>50000$, since there is enough file replication in the network, the load balancing performance is greatly improved due to power of two choices. This is in accordance with the lessons learned from \S\ref{sec:Proximity-Aware_Two_Choices_Strategy}. Also, for $10000 < n < 50000$, we have a transition region where a mixed behaviour is observed. 
Likewise, the curve for $M=2$ shows a similar trend. 
However, for $M=10$ due to memory abundance, we only observe the latter behaviour where power of two choices is achieved.
Observations made above from Figure~\ref{fig:TwoChoice_SrvSzVar_MaxLoad} has an important practical implication. Since employing Strategy~II is only beneficial in networks with high file replication, for other situations with limited cache size, the less sophisticated Strategy~I is a more proper choice.
 
Figure~\ref{fig:TwoChoice_SrvSzVar_AvgCost} draws the communication cost versus number of servers for various cache sizes for similar setting used in Figure~\ref{fig:TwoChoice_SrvSzVar_MaxLoad}. Since in this figure there is no constraint on the proximity the communication cost growth is of order $\Theta(\sqrt{n})$.

In simulations depicted in Figures~\ref{fig:TwoChoice_SrvSzVar_MaxLoad}~and~\ref{fig:TwoChoice_SrvSzVar_AvgCost}, we only consider the case $r=\infty$. In order to investigate the effect of parameter $r$ on the performance of the system, in Figure~\ref{fig:Tradeoff}, we have simulated network operation for different values of $r$. This results in a trade-off between the maximum load and communication cost, as shown in Figure~\ref{fig:Tradeoff}.  Here we consider a torus with $2025$ servers, where $500$ files with Uniform popularity are placed uniformly at random in each node. Each point is an average of $5000$ simulation runs.

In this figure, like before, \ie, Figure~\ref{fig:TwoChoice_SrvSzVar_MaxLoad}, we observe two performance regimes based on the cache size $M$. In high memory regime, \eg, for curves corresponding to $M=50$ and $M=200$, we can achieve the power of two choices by sacrificing a negligible communication cost. On the other hand, in low memory regime, \ie, $M=1$, we cannot decrease the maximum load even at the expense of high communication cost values. For intermediate values of $M$, we clearly observe the trade-off between the maximum load and communication cost.

%For Strategy~II, we also run some simulations that are shown in Figures~\ref{fig:TwoChoice_SrvSzVar_AvgCost}~and~\ref{fig:TwoChoice_SrvSzVar_MaxLoad}. Similar to Strategy~I, the average cost and maximum load for various cache size are plotted versus network size which varies from $2025$ to $100489$ nodes.
%Figure~\ref{fig:TwoChoice_SrvSzVar_MaxLoad} shows an interesting behaviour of maximum load under strategy~II. As our theoretical results predicts when the number of servers grow, we have to get the power-of-two-choices behaviour. However, as confirmed by Figure~\ref{fig:TwoChoice_SrvSzVar_MaxLoad}, the size of network at which one gets the effect of power-of-two-choices depends on the cache size at each server. In fact the first part of curves for cache size $1$, $2$, and $5$ in Figure~\ref{fig:TwoChoice_SrvSzVar_MaxLoad} represent a one-choice behaviour similar to Figure~\ref{fig:OneChoice_SrvSzVar_MaxLoad}.

%Finally Figure~\ref{fig:Tradeoff} shows the trade-off between the maximum load and the average cost for various cache size when the Strategy~III is used for redirecting the incoming requests. Here the number of servers is $2025$ and the library size is $500$.

\begin{figure}
\begin{center}
\includegraphics[width=0.48\textwidth]{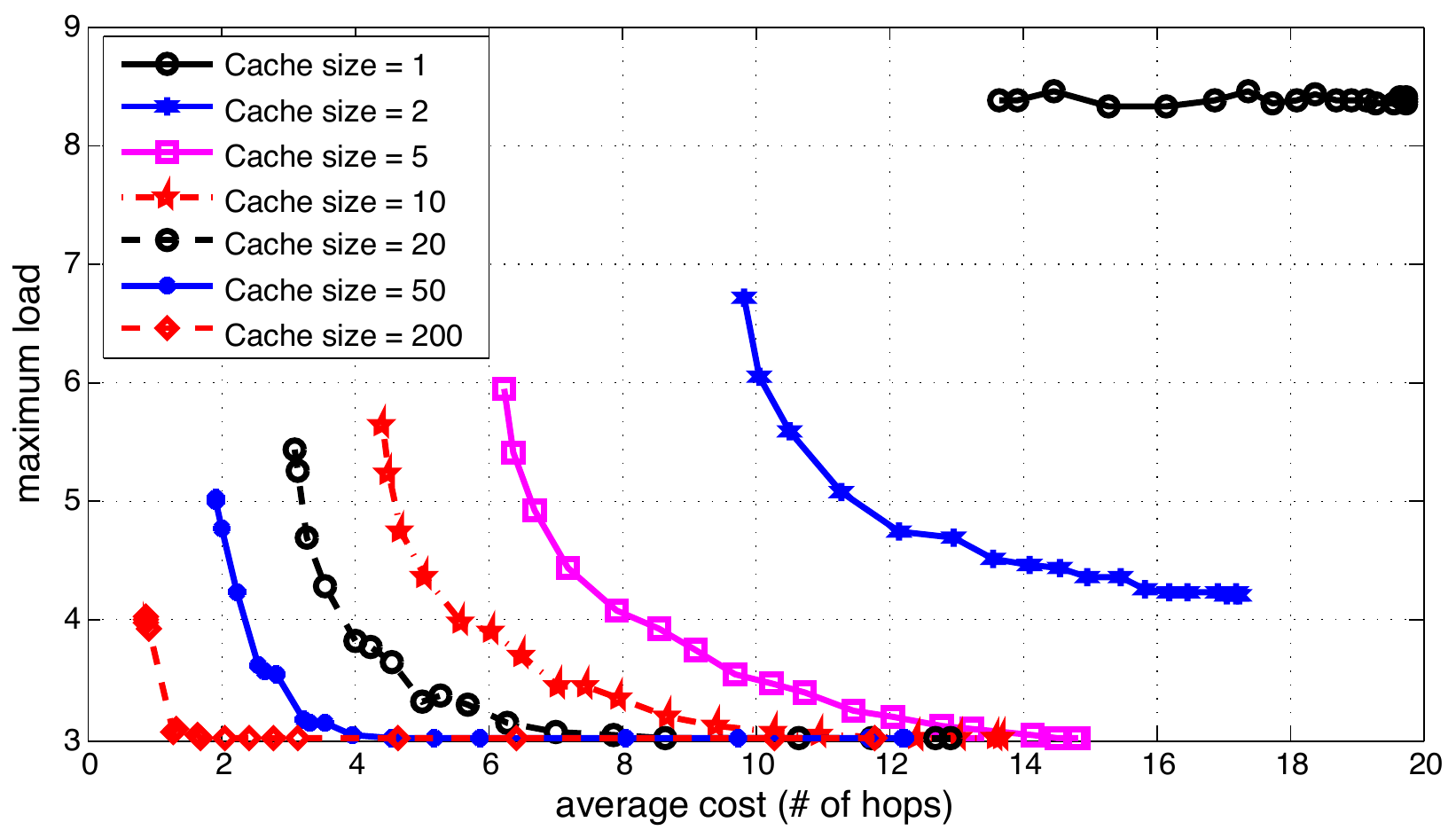}
\end{center}
\caption{The tradeoff between the maximum load and communication cost for Strategy~II.}\label{fig:Tradeoff}
\end{figure}

%\subsection{Power of Two Choices Load Balancing Strategy}

%\subsection{Communication Cost versus Load Balancing}

%\subsection{More Practical Scenarios: Queueing Models, Other Graph Topologies}

%% file: 9-Discussions.tex
\section{Discussion, Open Questions and Future Directions}\label{sec:disscuss}
In this section, first, we summarize the paper. Then we bring forward discussion about the proposed schemes, open questions and possible future directions.

In summary, we have considered the problem of randomized load balancing and its tension with communication cost in cache networks. By proposing two request assignment schemes, the trade-off between communication cost and maximum load has been investigated analytically. Moreover, simulation results support our theoretical findings and provide practical design guidelines.

%\emph{(Number of Request vs. Number of Servers)}: 
%In this paper, we have assumed that the total number of requests and the number of servers are equal to $n$. In a more general setting assume that the number of requests is $m$, not necessarily equal to the number of servers $n$. It can be easily seen that all the results hold true for the case of $m=\Theta(n)$ as well. If we have $m=\omega(n)$ then since we have much more requests than the servers, the deviation from the average load at each server, i.e. $m/n$, is negligible and the maximum load will be $L=\Theta(m/n)$. On the other hand if $m=o(n)$ then it can be easily seen that the maximum load is $L=\Theta(1)$. Thus the only non-trivial case to be addressed is the case of $m=\Theta(n)$ which is thoroughly discussed inside the paper.

%\emph{(Distributed Implementation)}: 
%Also notice that 
The proposed \emph{proximity-aware two choices} scheme can be implemented in a distributed manner. To see why, notice that upon arrival of each request at each server, this strategy needs two kinds of information to redirect the request. This information can be provided to the requesting server without the need for a centralized authority in the following way. The first one is the cache content of other users in its neighborhood with radius $r$. Since, the cache content dynamic of servers is much slower than the requests arrival, this can be done by periodic polling of nearby servers without introducing much overhead. Also, the cache content placement at each server can be implemented via efficient Distributed Hash Table (DHT) schemes (see, \eg, \cite{BauerHW07} and \cite{KargerLLPLL97}), which can be adopted to dynamic library popularity profiles. This will also enable all users to obtain global cache content information  in a robust and distributed manner. In this paper we assume a static profile and do not go into the details of such schemes. The second type of information is the queue length information of two randomly chosen nodes inside its neighborhood with radius $r$, which can also be efficiently done in a distributed manner by polling or piggybacking.

%\emph{(Relation to the Supermarket Model)}:
In practice, request arrivals and servers' operation happen in continuous time which needs a queuing theory based performance analysis. However, as shown in \cite{Mitzenmacher01} and \cite{sur01}, the behaviour of load balancing schemes in continuous time (\ie, known as the supermarket model) and static balls and bins problems are closely related. Thus, we conjecture that our proposed scheme will also have the same performance in queuing theory based model. We postpone a rigorous analysis of such scenario to future work.

In this paper we do not consider any form of coding in the cache content placement and content delivery phases. However, as recently shown in \cite{DBLP_Maddah-AliN14} (and follow up works \cite{DBLP_Caire_JiCM16,DBLP_Karamchandani_NM16,DBLP_Shariatpanahi_MK15}), employing coding in cache networks can reduce network traffic dramatically. An important future work will be investigating the effect of coding techniques in the context of our proposed randomized load balancing scheme.

%\mahdic{Add more discussions and future directions! Any idea?!}

%% file: 8-Appendix.tex
\section{Some Tail Bounds}\label{app:bounds}

\begin{thm}[Chernoff Bounds]\label{app:cher}
	Suppose that $X_1, X_2,\ldots, X_n\in \{0, 1\}$ are independent random variables and let $X=\sum_{i=1}^n X_i$. 
	Then for every $\delta\in (0, 1)$ the following inequalities hold:
	\begin{align*}
	\Pr{X\ge (1+\delta)\Ex{X}} &\le \exp(-\delta^2 \Ex{X}/2),\\
	\Pr{X\le (1-\delta)\Ex{X}} &\le \exp(-\delta^2 \Ex{X}/3).
	\end{align*}
	In particular,
	\[
	\Pr{|X-\Ex{X}|\ge \delta\Ex{X}}\le 2\exp(-\delta^2 \Ex{X}/3).
	\]
\end{thm}
For a proof see \cite{DP09}.

To deal with moderate independency we have the following lemma whose proof can be found in \cite[Lemma 1.18]{doerrbook}. 
\begin{lemma}[Deviation bounds for moderate independency]\label{mod-chernoff}
	Let $X_1,\cdots, X_n$ be arbitrary binary random variables. Let
	$X_1^*, X_2^*, \cdots, X_n^*$
	be binary random variables that are mutually independent and
	such that for all $i$, $X_i$,
	is independent of $X_1, \cdots, X_{i-1}$. Assume that for all
	$i$ and all $x_1, . . . , x_{i-1} \in \{0, 1\}$,
	\[
	\Pr{X_i = 1|X_1= x_1,\cdots, X_{i-1}= x_{i-1}}\ge \Pr{X^*_i= 1}.
	\]
	Then for all $k \ge 0$, we have
	\[
	\Pr{ \sum_{i=1}^nX_i\le k}\le \Pr{\sum_{i=1}^n X^*_i\le k}
	\]
	and the latter term can be bounded by any deviation bound for independent
	random variables.
\end{lemma}

\section{Proofs}\label{apndx:proofs}
\begin{proof}[Proof of Lemma~\ref{WindMillLemma}]
\textbf{Upper Bound -- } 
Fix a node $u$ and w.l.o.g. assume that 
$u$ is denoted by pair $(0,0)$ in the torus.
With respect to $u$ and some positive number $r>0$  define four areas as follows
\begin{align*}
A_1(u)&:=\{(x, y) : 0\le y \le x/2  \text{ and }  (x, y)\in B_{r}(u)\},\\
A_2(u)&:=\{(x, y): 0\le -x\le y/2
\text{ and }  (x, y)\in B_{r}(u)
\},\\
A_3(u)&:=\{(x, y): 0\le -y\le -x/2 
\text{ and }  (x, y)\in B_{r}(u)
\},\\
A_4(u)&:=\{(x, y): 0\le x\leq -y/2
\text{ and }  (x, y)\in B_{r}(u)
\},
\end{align*}
which are shown in Fig. \ref{Upper}. It is easy to see that all four areas have the same size, that is 
\begin{align}\label{eq:lower}
|A_1(u)| &= \sum_{y=0}^{\lfloor r/3 \rfloor}\sum_{x=2y}^{r-y}1 \nonumber\\
&=\sum_{y=0}^{\lfloor r/3\rfloor}(r-3y+1) \nonumber\\
&\geq \sum_{y=0}^{\lfloor r/3\rfloor} 3y  \nonumber\\
&\ge r^2/8.
\end{align}

Let us fix some arbitrary  $1\le j\le K$ 
and for every node $u$ define indicator random variable  $X_{u, j}$ taking value $1$ if  $u$ has cached file $W_j$ and there is no node in $A_1(u)$ that has cached file $W_j$, and $0$ otherwise. Then,
\begin{align*}
\Pr{X_{u, j}=1}=
\left(1-\left(1-\frac{1}{K}\right)^M\right)\left(1-\frac{1}{K}\right)^{M(|A_1(u)|-1)},
\end{align*}
where the first term determines the probability that $u$ caches $W_j$ and the second one determines the probability that nodes in $A_1(u)\setminus\{ u\}$ do not cache $W_j$. 
By setting $r=5\sqrt{K\log n/M}$ and applying Inequality (\ref{eq:lower}) we have, 
\begin{align}\label{approx1}
\left(1-\frac{1}{K}\right)^{M(|A_1(u)|-1)} &= \mathrm{e}^{-\frac{25\log n}{8}+M/K}(1+o(1)) \nonumber\\
&=O(n^{-3}),
\end{align}
where it follows from \mbox{$1-1/K= \mathrm{e}^{-1/K}(1+o(1))$} and $M/K=o(1)$. Moreover, by using the approximation \mbox{$1-(1-1/K)^M=M(1+o(1))/K$}, we have
\[
\Pr{X_{u,j}=1}\le  \frac{M(1+o(1))}{K\cdot n^3}.
\] 

Therefore applying the union bound over all $n$ nodes and $K$ files implies  that w.h.p. for every $u$ there exists at least one node in $A_1(u)$ which shares a common file with $u$, supported that we choose $r= 5\sqrt{K\log n/M}$. We can similarly prove the same argument for $A_2(u)$, $A_3(u)$ and $A_4(u)$. 

Suppose that $u$ has cached file $W_j$, and we want to find an upper bound for the size of the Voronoi cell centered at $u$ corresponding to $W_j$. In order to do this let us define $v^1=(v^1_x, v^1_y)\in A_1(u)$ to be the nearest node to $u$ with file $W_j$. Similarly define $v^i\in A_i(u)$, $2\le i\le 4$. W.l.o.g.
assume that $u$  is the origin and the nodes are in a $xy$-coordinate system. Define
 \[
B:=\{(x,y) : v^3_x\le x\le v^1_x \text{ and } v^4_y\le  y\le v^2_y\}.
\]
Now we show that the Voronoi cell of $u$ is contained in $B$, and thus the size of $B$ is an upper bound to the size of the Voronoi cell.
Consider Fig. \ref{Upper}. Let us consider  node  $w$ in the complement of $B$ with $w_y>v^2_y$ and $w_x>0$. Assume that  $P_{uw}$ is a shortest path from $w$ to $u$ that passes node $(0, v^2_y)$. By definition of $A_2(u)$, we know  the length of a shortest  path from    $(0, v^2_y)$ to $(v^2_x, v^2_y)$ is $|v^2_x|\leq v^2_y/2$. This shows that $w$ is closer to $v^2$ than $u$, and by definition it  does not belong to the Voronoi cell centered at $u$. We similarly can show that each arbitrary node $w\in B^c$ is closer to either $v^i$'s rather than $u$.  Since we arbitrarily choose $u$ and $1\le j\le K$, there is  sub-grid $B$ that contains every Voronoi cell in $\mathcal{V}_j$, centered at any given $u$.  
So the size of any Voronoi cell centered at an arbitrary node $u$ is bounded from above by $|B|$, which is at most $4r^2=O(K\log n/M)$.

\begin{figure}
\begin{center}
\includegraphics[width=0.5\textwidth]{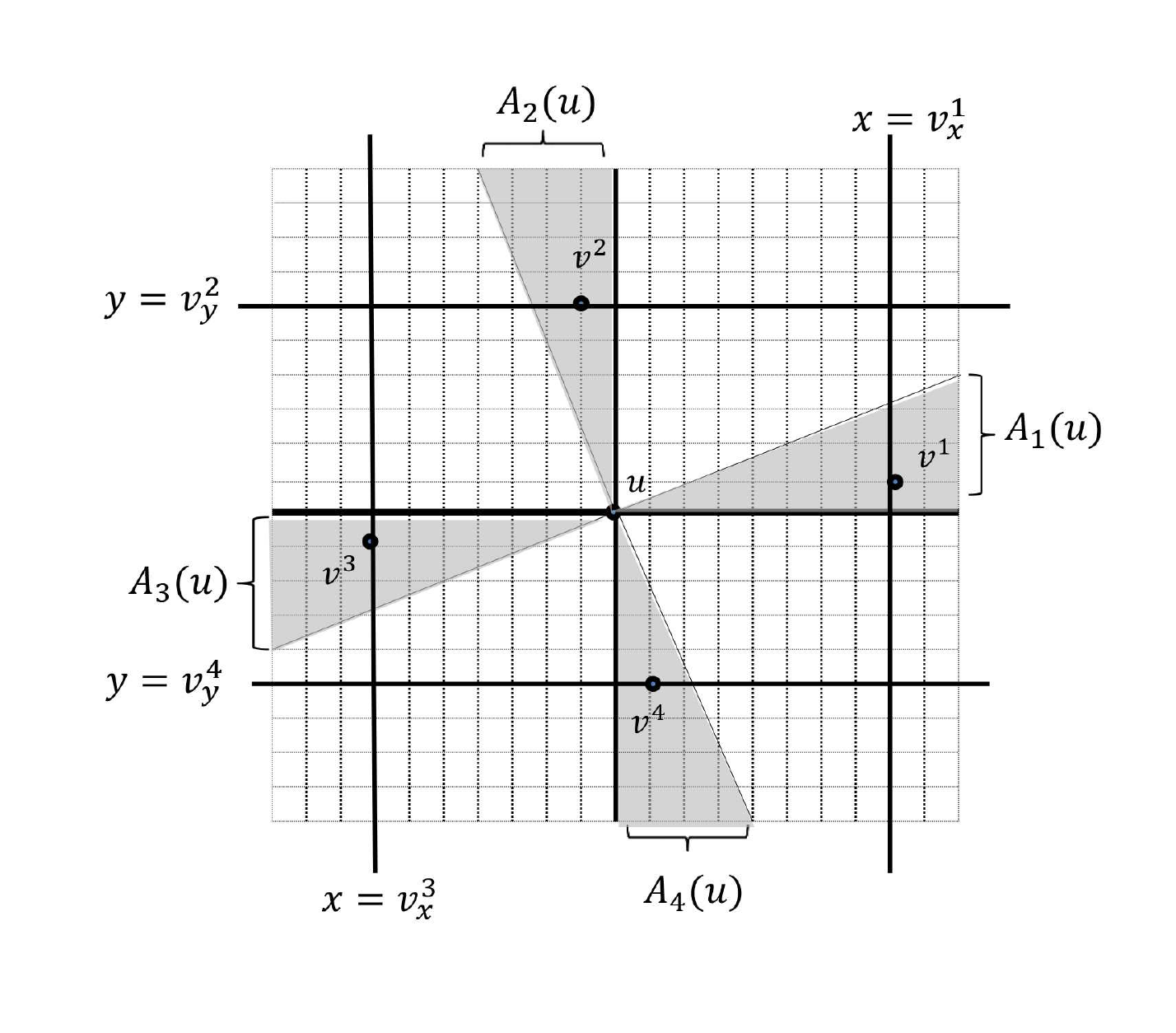}
\end{center}
\caption{Demonstration of regions $A_1(u),\ldots,A_4(u)$ used in the upper bound proof of Lemma~\ref{WindMillLemma}.}\label{Upper}
\end{figure}

\textbf{Lower Bound --} Let us define indicator random variable $Y_{u, j}$ for every $u$ and some fixed  $j$ taking value $1$ if $u$ has cached file $W_j$ and there is no  $v\in B_{r}(u)$ that has cached file $W_j$, and $0$ otherwise.	
Note that $|B_r(u)\setminus\{u\}\}|=2r(r+1)$.
Then we have  
\[
\Pr{Y_{u,j}=1}=\left(1-\left(1-\frac{1}{K}\right)^M\right)\left(1-\frac{1}{K}\right)^{M[2r(r+1)]}.
\]
By setting $r=\sqrt{{\epsilon\cdot K\cdot \log n /4M}}$ and using similar approximations used in (\ref{approx1}) we have
\[
p \triangleq \Pr{Y_{u,j}=1}=\frac{M(1+o(1))}{K\cdot n^{\epsilon/2}}.
\]
Let $Y_j=\sum_{u\in}  Y_{u, j}$. Then, we have the following claim.

\begin{clm}
For every $j$ we have $Y_j\ge 1$ with probability $1-o(1)$.
\end{clm}

This claim shows that there exists at least a Voronoi cell of size $\Theta(r^2)=\Theta(K \log n /M)$ which concludes the proof.

Now, in order to prove the claim note that
\[
\Ex{Y_j}=\sum_{u}\Ex{Y_{u,j}}=n\cdot\frac{M(1+o(1))}{K\cdot n^{\epsilon/2}}=Mn^{\epsilon/2}(1+o(1)).
\]
Also, we know that  
\begin{align}\label{var:main}
\Var{Y_j} &= \sum_{u, v}\mathrm{Cov}(Y_{u,j}, Y_{v,j}) \nonumber\\
&=\sum_{u, v}\Ex{Y_{u,j} Y_{v,j}} - \Ex{Y_{u,j}}\Ex{Y_{v,j}} \nonumber\\
&=\sum_{u, v} \Big( \Pr{Y_{u,j}=1, Y_{v,j}=1} \nonumber\\
&\quad\quad\quad -\Pr{Y_{u,j}=1}\cdot \Pr{Y_{v,j}=1}\Big),
\end{align}
%&=\sum_{u}\{Cov(Y_u, Y_u)+\sum_{v: 0<d(u, v)\le r}Cov(Y_u, Y_v)\\
%&+\sum_{v:r<d(u,v)\leq 2r}Cov(Y_u, Y_v)+\sum_{v: d(u, v)>2r}Cov(Y_u, Y_v)\}
where the last equality holds because 
$Y_{u, j}$'s are indicator random variables. 
It is easy to see that for every $u$
and $v$ with $d_G(u, v)>2r$, $\mathrm{Cov}(Y_{u,j}, Y_{v, j})=0$ as cache content placement at different nodes are independent processes. So we only consider pairs $u$ and $v$, with $d_G(u, v)\le 2r$.  Then for each pair of nodes three following cases should be considered:
\begin{itemize}
	\item $u=v$:
	In this case we have 
	\begin{align}\label{var:fisrt}
	\Pr{Y_{u,j}=1, Y_{v,j}=1}-\Pr{Y_{u,j}=1}\cdot \Pr{Y_{v,j}=1}\nonumber \\
	=\Pr{Y_{u,j}=1}-p^2=p(1-p).
	\end{align}
	
	\item $0<d_G(u, v)\le r$:
	By definition of indicator random variables $Y_{u, j}$'s, we have
	\begin{align}\label{var:sec}
	&\Pr{Y_{u,j}=1, Y_{v,j}=1}-\Pr{Y_{u,j}=1}\cdot \Pr{Y_{v,j}=1}\nonumber\\
	&=0-p^2. 
	\end{align}
	\item $r<d_G(u, v)\leq 2r$:
	In this case we have  	
	\begin{align}\label{var:thi}
	&\Pr{Y_{u,j}=1, Y_{v,j}=1}-\Pr{Y_{u,j}=1}\cdot \Pr{Y_{v,j}=1} \nonumber\\
	&=\Pr{Y_{u,j}=1|Y_{v,j}=1}\Pr{Y_{v,j}=1}-p^2 \nonumber \\ 
	&\leq \frac{M(1+o(1))}{K}p-p^2\leq \frac{2M}{K}p. 
	\end{align}
\end{itemize}

Now let us split the summation (\ref{var:main}) based on $d_G(u,v)$ as follows
\begin{align*}
\Var{Y_j} &=\sum_{u}\mathrm{Cov}(Y_{u,j},Y_{u,j})\\
&\quad+\sum_{u} \sum_{v: 0 < d_G(u, v)\leq r}\mathrm{Cov}(Y_{u,j}, Y_{v,j}) \\
&\quad+ \sum_{u} \sum_{v: r<d_G(u, v)\leq 2r}\mathrm{Cov}(Y_{u,j}, Y_{v,j}).
\end{align*}

Applying results (\ref{var:fisrt}-\ref{var:thi}) yields
\begin{align*}
\Var{Y_j} &\leq np(1-p)-n|B_r(u)|p^2+n|B_{2r}(u)|\frac{2Mp}{K} \nonumber\\
&\leq np+ 4nr(2r+1)\frac{2Mp}{K} \nonumber\\
&\leq np+6\epsilon np \log n \nonumber\\
&\leq 7np\log n,  
\end{align*}
where we use the fact that 
$4r(2r+1)\le 9r^2\leq 3\epsilon K \log n$.
Applying Chebychev's inequality  leads to 
\begin{align*}
\Pr{|Y_j-\Ex{Y_j}|\ge \Ex{Y_j}/2}&\le \frac{4\Var{Y_j}}{\Ex{Y_j}^2}\\
&\le \frac{28 np\log n}{n^2p^2}\\
&=\frac{28\log n}{np}=o(1).
\end{align*}
Therefore, $Y_j$ concentrates around its mean, \ie, $np = \Theta(n^{\epsilon/2})$, which proves the claim.
\end{proof}

%

%--------------------------------------------------------------------

\begin{proof}[Proof of Theorem \ref{thm:Strategy_I_Large}]
	To establish  upper bound $O(\log n)$ for the maximum load, we follow the first part of proof of Theorem \ref{thm:imblance}.
	To obtain a lower bound, consider an arbitrary server $u$ that has cached file set $S$ with  $s$ distinct files. Note that by  Lemma 
	\ref{lem:GoodnessProperty}, we have for every node $u$, $s=\Theta(M)$ with high probability.
	Let us define the indicator random variable
	$X_{u, j}$, $W_j\in S$, taking $1$ if the nearest replica of $W_j$ is outside of $B_r(u)$, where $r=\sqrt{K/2M}$  and zero otherwise. 
	It is easy to see that $X_{u,j}$'s are correlated. For example, consider set of files $T=\{W_{j_1}, W_{j_2},\ldots, W_{j_t}\}\subset S$, where $X_{u,j'}=1$ for every $W_j'\in T$. Then conditional on this event, each node in $B_r(u)$ has cached files from a subset of the library of size $K-|T|$.
	Then probability that a node in $B_r(u)$ caches $W_j$ is at most $M/(K-t)$.
	Hence,  for every $W_j\in S$,
	\begin{align*}
	&\Pr{X_{u,j}=1 |\{X_{u, j'}=x_{u,j'}, W_{j'}\in S\setminus\{W_j\}\}}\\
	&\geq\left(1-\frac{M}{K-\sum_{ W_{j'}\in S\setminus\{W_j\}}x_{u, j'}}\right)^{2r(r+1)}\\
	&\geq \left(1-\frac{M}{K-M+1}\right)^{2r(r+1)}=\mathrm{e}^{-\Omega(1)}=p,
	\end{align*}
	where $|B_r(u)\setminus\{u\}|=2r(r+1)=\Theta(K/M)$ and hence $p$ is a constant.
	Let $Z=\sum_{W_j\in S}X_{u,j}$ and then $\Ex{Z}\geq s\cdot p$. Using a Chernoff bound for moderately correlated indicator random variables (e.g., see Lemma \ref{mod-chernoff} ) implies that 
	\[
	\Pr{Z<sp/2}=o(1/n^2).
	\] 
	Therefore $B_r(u)$ does not contain any replica of at least $p/2$ fraction of files cached at $u$, namely $S'$. Using the union bound over all nodes we deduce the similar statement for every node with probability at least $1-o(1/n)$.
	Therefore, for every $u$ we have,
	\[
	\Pr{\text{$u$ severs a request}} \geq \frac{|B_{r/2}(u)|}{n}\cdot \frac{|S'|}{K}=\Omega(1/n)
	\]
	where it follows from $|S'|=\Theta(M)$, $|B_{r/2}(u)|=\Theta(K/M)$.
	Since there are $n$ requests, it is easy to see that   the load of each server is bounded from below by a  Poisson distribution $\mathrm{Po}(c)$, where $c$ is a constant. 
	On the other hand, it is known that (e.g., see \cite{Devroye85} ) the maximum number taken by  $n$ Poisson distribution $\mathrm{Po}(c)$ is $\Omega(\log n/\log\log n)$ w.h.p. and hence the lower bound is proved.

\end{proof}

\begin{proof}[Proof of Theorem~\ref{thm:Strategy_I_CommCost}]	
	Assume a request from an arbitrary node $u$ for file $W_j$. The probability that this file is cached at another node $v$ is $q_j :=1-(1-p_j)^M$. Cache content placement at different nodes is independent. Thus, the  number of nodes which should be probed is a geometric random variable with success probability $q_j$. This results in the average $1/q_j$ trials that leads to expected distance of 
	\begin{equation}
	\Theta\left(\frac{1}{\sqrt{q_j}}\right)=\Theta\left(\frac{1}{\sqrt{1-(1-p_j)^M}}\right).
	\end{equation}
	When averaged over different files we will have
	\begin{equation}
	C=\sum_{j=1}^K {p_j\Theta\left(\frac{1}{\sqrt{1-(1-p_j)^M}}\right)}.
	\end{equation}
	\begin{itemize}
		\item For uniform distribution we have $p_j=1/K$ and then
		\begin{equation}
		C=\Theta(\sqrt{K/M}).
		\end{equation}
		\item For Zipf distribution with $M=\Theta(1)$ we have
		\begin{align}\label{final}
		C &= \sum_{j=1}^K {p_j\Theta\left(\frac{1}{\sqrt{1-(1-p_j)^M}}\right)} \nonumber \\
		&= \sum_{j=1}^K p_j \Theta \left( \frac{1}{\sqrt{p_j M}} \right) \nonumber \\ 
		&= \Theta\left(      		\frac{\sum_{j=1}^Kj^{-\gamma/2}}{\left(M\sum_{j=1}^K j^{-\gamma}\right)^{1/2}}  \right). \nonumber \\
		%	&:= \Theta\left(      		\frac{\Lambda(\beta/2)}{\sqrt{M\Lambda(\beta)}}  \right),
		\end{align}
	\end{itemize}
	Define $\Lambda(\gamma) :=\sum_{j=1}^Kj^{-\gamma}$, for every $\gamma$.		
	On the other hand  it is known that  for every $\gamma>0$ (\eg, see \cite{JiTLC15})
	
	\iffalse
	\begin{equation}
	\left\{
	\begin{array}{llll}
	\Lambda(\gamma/2)=\Theta\left(K^{1-\gamma/2}\right), &\Lambda(\gamma)=\Theta\left(K^{1-\gamma}\right) &: & \quad 0 < \gamma<1, \\
	\Lambda(\gamma/2)=\Theta\left(K^{1/2}\right), &\Lambda(\gamma)=\Theta\left(\log K\right) &: & \quad  \beta=1, \\
	\Lambda(\gamma/2)=\Theta\left(K^{1-\gamma/2}\right), &\Lambda(\gamma)=\Theta(1) &: & \quad  1< \gamma <2, \\ 
	\Lambda(\gamma/2)=\Theta\left(\log K\right), &\Lambda(\gamma)=\Theta(1) &: &\quad \gamma=2, \\
	\Lambda(\gamma/2)=\Theta(1), &\Lambda(\gamma)=\Theta(1) &: &\quad \beta>2.
	\end{array}
	\right.	
	\end{equation} 	
	\fi
	
	\begin{equation}\label{zipfestimate}
	\Lambda(\gamma)=\left\{
	\begin{array}{ll}
	\Theta\left(K^{1-\gamma}\right),  & \quad 0 < \gamma <1, \\
	\Theta\left(\log K\right), &\quad \gamma=1, \\
	\Theta(1), &\quad \gamma>2.
	\end{array}
	\right.	
	\end{equation}

	\iffalse
	\begin{lemma}\label{lem:Zipf_Bound}
		If $\beta\neq 1$, then we have
		\begin{align*}
		& \frac{1}{1-\beta}\left[ (k_2+1)^{1-\beta} - k_1^{1-\beta}  \right]  \le \sum_{j=k_1}^{k_2} j^{-\beta} \\
		& \le \frac{1}{1-\beta} \left[ k_2^{1-\beta} - k_1^{1-\beta} \right] + \frac{1}{k_1^{\beta}}.
		\end{align*}
	\end{lemma}
	\fi
	
	Now inserting the above equations  into (\ref{final}) completes the proof. 
\end{proof}

\section{Proof of Example \ref{ex:prop_M1}} \label{apndx:Examp-M1}

\begin{proof}%[Proof of Example~\ref{ex:prop_M1}]
	It is easy to see that for $M=1$, the number of caching servers with a specific file, say $W_j$ denoted by $S_j$, is distributed as a $\mathrm{Bin}(n, p_j)$. Thus applying a  Chernoff bound for $S_j$ (\eg, see Appendix~\ref{app:bounds}) implies that 
	\[
	\Pr{|S_j-\Ex{S_j}|\geq \Ex{S_j}/2}\le 2\exp({-p_j n/12}).
	\]
	Moreover, let $R_j$ denote the number of requests for file $W_j$, which is the sum of $n$ i.i.d. $\mathrm{Bin}(n, p_j)$ random variables. Again applying a Chernoff bound (\eg, see Appendix~\ref{app:bounds}) for Poisson random variables yields that
	\[
	\Pr{|R_j-\Ex{R_j}|\geq \Ex{R_j}/2}\le 2\exp({-p_j n/12}).
	\]
	Notice that $\Ex{S_j}=\Ex{R_j}=np_j$. 
	Suppose that  $\mathcal{A}_j$ denotes the event that $|S_j-\Ex{S_j}|\leq \Ex{S_j}/2$ and $|R_j-\Ex{R_j}|\leq \Ex{R_j}/2$. Then we have that $\Pr{\mathcal{A}_j}\ge 1-4\exp(-p_jn/12)$. Also define $\mathcal{E}_j$  to be  the event that two-choice model with $S_j$ bins (caching servers) and $R_j$ balls (requests) achieves maximum load $R_j/S_j+\Theta(\log\log S_j)$. It is shown that this event happens with probability  $1-O(1/S_j^c)$, for every constant $c$ (e.g., see \cite{ABKU99}).
	So we have that 
	\begin{align*}
	\Pr{\mathcal{E}_j}&=\Pr{\mathcal{E}_j| \mathcal{A}_j}\Pr{\mathcal{A}_j}+
	\Pr{\mathcal{E}_j| \neg\mathcal{A}_j}\Pr{\neg\mathcal{A}_j}\\
	&>(1-2(p_j n)^{-c})(1-4\exp(-p_jn/12))\\
	&+ (1-2(p_j n)^{-c})(4\exp(-p_jn/12))\\
	&\ge 1-8(p_j n)^{-c}.
	\end{align*}
	Since we have $K$ disjoint subsystems, the union bound over all subsystems shows that the two choice model does achieve  the desired maximum load with probability 
	$1-8\sum_{j=1}^K(np_j)^{-c}=1-o(1)$ which concludes the proof due to example's assumption on popularity profile. 
	
	Now we show that the Uniform and Zipf distributions satisfy the example's assumption. When $\mathcal{P}$ is the Uniform distribution over $K$ files, $\forall j,\ p_j\cdot n=n^{\epsilon}$. Now by setting $c=3/\epsilon$, we have that 
	\[
	\sum_{j=1}^K(np_j)^{-c}=K({1}/{n^\epsilon})^c=K/n^3=o(1/n^2).
	\]
	
	Also, for Zipf distribution we have
	\[
	p_j=\frac{j^{-\gamma}}{\sum_{j=1}^K j^{-\gamma}}=\frac{j^{-\gamma}}{\Lambda(\gamma)}.
	\]
	\iffalse
	So we have that 
	\begin{align*}
	\sum_{j=1}^K(np_j)^{-c}
	&=\left(\frac{\Lambda(\beta)}{n}\right)^c\Lambda(\beta c).
	\end{align*}
	It is know that  (e.g., see \cite{}) 
	\begin{equation*}
	\Lambda(\gamma)=\left\{
	\begin{array}{ll}
	\Theta\left(K^{1-\gamma}\right),  & \quad 0 < \gamma <1, \\
	\Theta\left(\log K\right) &\quad \gamma=1, \\
	\Theta(1) &\quad \gamma>2.
	\end{array}
	\right.	
	\end{equation*} 
	\fi	
	Depending on $\gamma$,  we consider two cases,
	\begin{itemize}
		\item $\gamma\ge 1$: For every $c>1$ we have
		\[
		\left(\frac{\Lambda(\gamma)}{n}\right)^c\Lambda(\gamma c)=\Theta\left(\frac{\log^c K}{n^c}\right)\Lambda(\gamma c)=o(1),
		\]
		where we used $K<n$ and Equality (\ref{zipfestimate}).
		\item $0< \beta<1$: By setting 	 $c=2/ \gamma$ and using the fact that $K<n$, we have
		\begin{align*}
		\left(\frac{\Lambda(\gamma)}{n}\right)^c
		\Lambda(\gamma c)=&\Theta\left( \frac{K^{(1-\gamma)c}}{n^c}\right)
		\le {n^{(1-\gamma)c-c}}\Lambda(\gamma c)\\
		=&n^{-\gamma c}\Lambda(\gamma c)=o(1),
		\end{align*}
		where we applied Equality (\ref{zipfestimate}).
	\end{itemize}

\end{proof}

%---------------------------------------------------------------------
%---------------------------------------------------------------------